\documentclass[letterpaper]{article} % DO NOT CHANGE THIS
\usepackage{aaai24}  % DO NOT CHANGE THIS
\usepackage{times}  % DO NOT CHANGE THIS
\usepackage{helvet}  % DO NOT CHANGE THIS
\usepackage{courier}  % DO NOT CHANGE THIS
\usepackage[hyphens]{url}  % DO NOT CHANGE THIS
\usepackage{graphicx} % DO NOT CHANGE THIS
\usepackage{natbib}  % DO NOT CHANGE THIS AND DO NOT ADD ANY OPTIONS TO IT
\usepackage{caption} % DO NOT CHANGE THIS AND DO NOT ADD ANY OPTIONS TO IT
\usepackage{algorithm}
\usepackage{algorithmic}

\usepackage{newfloat}
\usepackage{listings}
\DeclareCaptionStyle{ruled}{labelfont=normalfont,labelsep=colon,strut=off} % DO NOT CHANGE THIS
\floatstyle{ruled}
\newfloat{listing}{tb}{lst}{}
\floatname{listing}{Listing}
\title{Using Symmetries to Lift Satisfiability Checking\thanks{This research received funding from the Flemish Government under the ``Onderzoeksprogramma Artificiële Intelligentie (AI) Vlaanderen'' programme.}}
\author {
    Pierre Carbonnelle\textsuperscript{\rm 1},
    Gottfried Schenner\textsuperscript{\rm 2},
    Maurice Bruynooghe\textsuperscript{\rm 1},
    Bart Bogaerts\textsuperscript{\rm 3},
    Marc Denecker\textsuperscript{\rm 1}
}
\affiliations {
    \textsuperscript{\rm 1}KU Leuven, Belgium\\
    \textsuperscript{\rm 2}Siemens, Austria\\
    \textsuperscript{\rm 3}Vrije Universiteit Brussels, Belgium\\
    \{pierre.carbonnelle, maurice.bruynooghe, marc.denecker\}@kuleuven.be, \\gottfried.schenner@siemens.com, bart.bogaerts@vub.be
}

\usepackage[usenames,dvipsnames]{xcolor}
\usepackage{amsmath,amssymb}
\usepackage{booktabs, threeparttable} % <-- packages used in table
\usepackage{mathabx}
\usepackage{numdef}
\usepackage{pifont}
\usepackage{ stmaryrd }
\usepackage{subcaption}
\usepackage{tikz}
\usepackage{todonotes}
\usepackage{xparse}
\usepackage{xspace}

\renewcommand{\tnote}[1]{\textsuperscript{#1}}

\usepackage{{makecell}}
\usepackage{tabularx}
\usepackage{array,collcell}
\newcommand\AddLabel[1]{%
  \refstepcounter{equation}% increment equation counter
  (\theequation)% print equation number
  \label{#1}% give the equation a \label
}
\newcolumntype{M}{>{\hfil$\displaystyle}X<{$\hfil}} % mathematics column
\newcolumntype{L}{>{\collectcell\AddLabel}r<{\endcollectcell}}

\newcommand{\advanced}[1]{} % {{\color{green} #1}}
\newcommand{\red}[1]{{\color{red}#1}}

\newcommand\m[1]{\ensuremath{#1}\xspace}
\newcommand{\ignore}[1]{}

\newcommand{\free}{\m{\mathit{vars}}}
\newcommand{\mpc}[1]{\m{\mathit{#1}}}

\newcommand{\Hole}{\m{\mathit{Hole}}}
\newcommand{\holeOf}{\m{\mathit{holeOf}}}

\newcommand{\Pigeon}{\m{\mathit{Pigeon}}}
\newcommand{\isIn}{\m{\mathit{isIn}}}

\newcommand{\vars}{\m{\mathit{vars}}}

\newcommand{\F}{\m{\phi}}

\newcommand{\RC}{\m{\mathit{RC}}}
\newcommand{\Th}{\m{\mathcal{T}}}
\newcommand{\ConfigObject}{\m{\mathit{ConfigObject}}}

\newcommand{\Tr}{\textbf{true}} % {\m{\top}}
\newcommand{\I}{\m{\mathit{I}}}
\newcommand{\V}{\m{\mathit{\Sigma}}}

\newcommand{\Exp}{\m{\mathit{exp_{\pi}}}}
\newcommand{\mul}{\m{\mathit{mul}}}
\newcommand{\lcm}{\m{\mathit{lcm}}}
\newcommand{\LCM}{\m{\mathit{Lcm}}}

\newcommand{\MUL}{\m{\mathit{Mul}}}
\newcommand{\Mul}{\m{\mathit{Mul}}}

\newcommand{\tr}{\chi}
\newcommand{\ch}[1]{\m{\chi(#1)}} %  or {\m{\chi_{#1}}}

\newcommand{\SUM}[3]{\m{\sum_{\begin{subarray}{c}#1:#2\end{subarray}}#3}}
\newcommand{\SUMW}[3]{\m{\sum_{#1}(#3 \mpc{~if~}#2)}}

\newcommand{\bard}{\m{\bar{d}}}
\newcommand{\barl}{\m{\bar{l}}}
\newcommand{\bart}{\m{\bar{t}}}
\newcommand{\barT}{\m{\bar{T}}}
\newcommand{\barx}{\m{\bar{x}}}

\newcommand{\mydots}{\makebox[1em][c]{.\hfil.\hfil.}}
\newcommand{\tuple}[2]{\m{#1, \mydots, #2}}

\newcommand{\fota}{FO(Type, Aggregate)}

\newtheorem{definition}{Definition}
\newtheorem{example}{Example}
\newtheorem{proof}{Proof}
\newtheorem{proposition}{Proposition}
\newtheorem{theorem}{Theorem}
\newtheorem{lemma}{Lemma}

\makeatletter
\newcommand\footnoteref[1]{\protected@xdef\@thefnmark{\ref{#1}}\@footnotemark}
\makeatother

\usepackage[usenames,dvipsnames]{xcolor}
\usepackage{xspace}

\def \ifempty#1{\def\temp{#1} \ifx\temp\empty }

\usepackage{xparse}

\NewDocumentCommand{\advice}{om}{%
\par {\color{OliveGreen} \textbf{Advice:} #2  
\ifempty{#1} {} \else {\bf In this case: #1} \fi 
} \par
}

\newsavebox\VerbTIPcommandforconsistentformatting
\begin{lrbox}{\VerbTIPcommandforconsistentformatting}
\begin{minipage}{\textwidth}
\begin{verbatim}
\newcommand\uriname[1]{\texttt{#1}\xspace}
\end{verbatim}
\end{minipage}
\end{lrbox}

\newsavebox\VerbTIPnobracketsarounddefinitionsWRONG
\begin{lrbox}{\VerbTIPnobracketsarounddefinitionsWRONG}
\begin{minipage}{\textwidth}
\begin{verbatim}
 \newcommand{\leqp}{{\leq_p}}
 \end{verbatim}
\end{minipage}
\end{lrbox}
\newsavebox\VerbTIPnobracketsarounddefinitionsRIGHT
\begin{lrbox}{\VerbTIPnobracketsarounddefinitionsRIGHT}
\begin{minipage}{\textwidth}
\begin{verbatim}
 \newcommand{\leqp}{\leq_p}
\end{verbatim}
\end{minipage}
\end{lrbox}

\begin{document}

\maketitle

\begin{abstract}
    \ignore{
        Generative configuration problems are configuration problems in which the number of some components in a valid configuration is not known prior to solving the problem.
        They are typically solved by step-wise increasing the assumed number of such components until a
        solution is found by non-generative configuration solvers.
        Solutions to such problems often have components that are
        interchangeable. They are the source of redundancies in the
        search space. The standard approach to reduce the search space, and thus improve search performance, is to add symmetry breaking rules.
        % We observe that a model of the standard pigeonhole configuration problem can be
        % compressed in a (so called) lifted model with a single pigeon and a
        % single hole and that the solution of the original problem can be
        % obtained by expanding the lifted model.
    }

We analyze how symmetries can be used to compress structures (also known as interpretations) onto a smaller domain without loss of information. This analysis suggests the possibility to solve satisfiability problems in the compressed domain for better performance. Thus, we propose a 2-step novel method: (i) the sentence to be satisfied is automatically translated into an equisatisfiable sentence over a ``lifted'' vocabulary that allows domain compression; (ii) satisfiability of the lifted sentence is checked by growing the (initially unknown) compressed domain until a satisfying structure is found.
The key issue is to ensure that this satisfying structure can always be expanded into an uncompressed structure that satisfies the original sentence to be satisfied.

We present an adequate translation for sentences in typed first-order logic extended with aggregates. Our experimental evaluation shows large speedups for generative configuration problems.  The method also has applications in the verification of software operating on complex data structures.  Our results justify further research in automatic translation of sentences for symmetry reduction.
%Further refinements of the translation are left for future work.

   %  A methodology to solve such problems is to make an assumption on the number of these components, to use standard configuration problem solvers, and to revise the assumption up, until a solution is found.
   %  In such methods, each domain element in models of the formula represents one component of the configuration. Typically, several components of a configuration are interchangeable, and large sets of interchangeable components can negatively impact the performance of model search. To improve this performance, we explore formulations in which a single domain element represents \emph{several} interchangeable components.  We call models of these formulations “lifted models”.

   % We describe a method to automatically transform the traditional formulation into a lifted formulation, and to automatically expand models of the lifted formulation into models of the traditional formulation.  We prove that the method is sound and complete for first-order logic extended with types and aggregates. Experiments show significant performance improvements on generative configuration problems found in the literature.
   % The method can be viewed as a new kind of symmetry elimination method.
\end{abstract}

\section{Introduction}

In made-to-order manufacturing, %a product is configured out of components to satisfy customer requirements.
the \emph{configuration problem} is the problem of finding a configuration of components
that satisfies the customer requirements and feasibility constraints~\cite{felfernig2014knowledge}.
Such problems can be solved by choosing a formal vocabulary and by representing the customer requirements and the feasibility criteria as a logic sentence to be satisfied.
A structure satisfying the sentence (a \textit{model}) represents an acceptable configuration.

Methods to solve configuration problems do not scale well, and various heuristics have been used to improve performance~\cite{schenner2017techniques}.
%Such problems often involve numeric aggregates, such as the total cost of the system or its power consumption.
Configurations often have components that are
interchangeable. They are the source of many redundancies in the
search space that negatively impact performance. The standard approach is to add symmetry breaking rules~\cite{DBLP:conf/kr/CrawfordGLR96}.
Here, we use another approach: we reformulate the problem to reduce symmetries~\cite{gent2006symmetry}.
This approach has been less studied, and is more an art than a science.

We observed that solutions to configuration problems can be
compressed in what we call a ``lifted model'', and that the solution of the original problem can be
obtained by expanding the lifted model.
This suggested that configuration problems, and, more generally, satisfiability problems, could be solved in the compressed domain: if the compressed domain is significantly smaller, it could lead to better performance.

%\todo{switch to predicate for isIn instead of function}
As a toy example, consider the following pigeonhole problem:
``given 10 pigeons and 5 pigeonholes, assign each pigeon to a pigeonhole such that each hole has (at most) 2 pigeons.''
With the appropriate vocabulary, a solution
is:
\begin{align}
    \Pigeon=\{&p_1,\ldots, p_{10}\} \label{ph:c1} \\
    \Hole=\{&h_1,\ldots, h_5\}\\
    \isIn=\{&(p_1, h_1), (p_2,h_2), \ldots, (p_5,h_5), \label{ph:c3} \\
        &(p_6,h_1), (p_7,h_2), \ldots, (p_{10},h_5)\} \nonumber
\end{align}
% determine the (minimum) number of pigeonholes such that each pigeon
% is allocated to one hole when each hole can hold at most 2 pigeons.

Note the symmetries:
another solution is obtained by exchanging 2
pigeons or 2 holes in the interpretation of $\isIn$.
% So, unless symmetry
% is broken, there is quite a bit of search before it is found that 4
% holes is not enough to obtain a solution.

This solution can be compressed to:
\begin{align}
    \Pigeon&=\{\mathbf{p}_1\} \label{ph:l1} \\
    \Hole&=\{\mathbf{h}_1\}\\
    \mul &= \{\mathbf{p}_1 \mapsto 10, \mathbf{h}_1 \mapsto 5\} \\
    \isIn&=\{(\mathbf{p}_1,\mathbf{h}_1)\} \label{ph:l4}
\end{align}
where the $\mul$ function indicates how many concrete domain elements each ``lifted'' domain element represents, and thus allows the domain compression.
The lifted pigeon (resp. hole) represents 10 concrete pigeons (resp. 5 concrete holes).

Even though the compressed structure only has 2 domain elements ($\mathbf{p}_1, \mathbf{h}_1$, excluding the naturals), it contains enough information to allow us to \emph{expand} it into a
model isomorphic to the original model. Following the theory we
develop in Section 3, the expansion is as follows:
\begin{align*}
    \mathbf{p}_1 & \leadsto p_1, \ldots, p_{10} \\
    \mathbf{h}_1 & \leadsto h_1,\ldots, h_5 \\
    (\mathbf{p}_1,\mathbf{h}_1) & \leadsto \{(p_1,h_1), \ldots, (p_5,h_5), (p_6,h_1), \ldots, (p_{10},h_5)\}
\end{align*}

%\bbcomment{Maybe it would be good to use different notation for LIFTED objects, eg $\mathbf{p}_0$ (now $p_0$ represents two different things)}
% of $lp_0$ is $p_0, \ldots, p_4$, the
% expansion of $lp_1$ is $p_5, \ldots, p_9$ and the expansion of $lh_0$ is
% $h_0,\ldots h_4$. The tuple $lp_0 \mapsto lh_0$ expands in the set of tuples
% $\{p_0 \mapsto h_0, p_1 \mapsto h_1, \ldots, p_4 \mapsto h_4\}$ and the tuple $lp_1 \mapsto lh_0$ expands in
% the set of tuples $\{p_5 \mapsto h_0, p_6 \mapsto h_1, \ldots, p_9 \mapsto h_4\}$.

% If the number of holes was fixed to 10, the compression
% would still be possible: there would be an additional lifted pigeonhole $lh_1$ representing 5 empty pigeonholes.

Furthermore, in Section~\ref{sec:lifted-theory}, we present a translation of sentences in typed first-order logic extended with aggregates into sentences that allow domain compression.
We show that a sentence in that language is equisatisfiable with its translation, given the number of concrete elements in each type: if either one is satisfiable, the other one is too.
This analysis allows us to solve satisfiability problems in the compressed domain. The novel method operates in two steps: (i) the sentence to be satisfied is automatically translated into a ``lifted'' sentence over the ``lifted'' vocabulary; (ii) satisfiability of the lifted sentence is checked by step-wise growing the compressed domain\footnote{An iterative method is required because the size of the compressed domain is not known in advance.} until a satisfying structure is found.
Crucially, this satisfying structure can always be expanded into a model of the original sentence.
%The crux of the method is that the original sentence and its translation must be equisatisfiable.
%We present in Section~\ref{sec:lifted-theory} a translation mechanism with that property.

Notice that the lifted model for a variation of the pigeonhole problem above with 100 times more pigeons and holes has the same number of lifted domain elements as the base case:
the multiplicity of each lifted domain element is simply multiplied by 100.
As a result, and unlike traditional symmetry breaking methods, our method solves that pigeonhole problem in constant time with respect to the domain size (excluding the naturals).
%\bbcomment{Techinically, this is not true. You cannot even READ the domain size in constant time. Maybe write something of the form ``in practice, this means that our method solves this ...'' or ``Meaning that the search space is indepent of the domain size}

We evaluate this method by comparing the time needed to find solutions for generative configuration problems discussed in the literature.
%Our evaluation shows that
Our method has significantly better performance than the traditional one for problems whose solution can be substantially compressed.
%\todo{+ set algebra}
%For some problems, lifted models are found in constant time with respect to domain size, while traditional methods scale at best linearly.

%In some configuration problems found in the literature, lifted models are found in constant time with respect to the input cardinalities, while traditional methods scale at best linearly.

Our paper is structured as follows:
after introducing our notation, we analyze how symmetries can be used to compress structures onto a smaller domain without loss of information (Section~\ref{sec:lifted-model}), and describe how to  lift a concrete sentence so that it is equisatisfiable with the lifted sentence (Section~\ref{sec:lifted-theory});
%we then prove the soundness and completeness of the approach (Section~\ref{sec:theory});
we describe the method, evaluate it on generative configuration problems (Section~\ref{sec:evaluation}),
and discuss applications in Boolean Algebra of sets with Presburger Arithmetic (BAPA) (Section~\ref{sec:BAPA})  before concluding with a discussion.

\section{Preliminaries}

This section introduces the logic language supported by our method, and the concept of permutation of a structure.

    % Table~\ref{tab:Legend} lists some of the symbols we use frequently.

    % \begin{table}
    %     \centering
    %     \begin{tabular}{ccc}
    %     \multicolumn{1}{c}{Category} & \multicolumn{1}{c}{Concrete symbol}  & \multicolumn{1}{c}{Lifted symbol}\\
    %     vocabulary  & \V & \VS    \\
    %     domain  & \D & \DA           \\
    %     domain element   & $d$ & $l$         \\
    %     model  & \I & \A
    %     \end{tabular}
    %     \caption{Legend}
    %     \label{tab:Legend}
    % \end{table}

\subsection{Typed first order logic with aggregates}
We call \fota~the language we support. %\bbcomment{Which language? This is the first sentence of the section}
We assume familiarity with first order logic~\cite{DBLP:books/daglib/0076838}.
A \emph{vocabulary} \V is a set of type, predicate, and function symbols.
Predicates and functions have a type signature (e.g. ${f:\barT\rightarrow T}$, where $\barT$ denotes a tuple of types).
Some symbols are pre-defined: type $\mathbb{B}$ (booleans), $\mathbb{Q}$ (rationals), equality, arithmetic operators and arithmetic comparisons.
\emph{Terms} and \emph{formulae} are constructed from symbols according to the usual FO syntactic rules.
We also allow sum aggregates (written as ${ \SUM{\barx \in \barT}{\phi}{t}}$ or as ${ \SUMW{\barx \in \barT}{\phi}{t}}$).
A cardinality aggregate $\#\{\barx \in \barT \mid \phi\}$ is a shorthand for a sum aggregate whose term $t$ is $1$.
Quantification and aggregation can only be over finite types.
Terms and formulas must be well-typed.
A formula without free-variable is called a \emph{sentence}.

A \emph{\V-structure} \I, consists of a domain and an \emph{interpretation} of each symbol of vocabulary \V.
% The interpretation $p^I$ of a predicate symbol $p : \barT\to\mathbb{B}$ (resp.\ $f^I$ of a function symbol $f:\barT \to T$) is a total function from $\barT^\I$ to $\mathbb{B}$ (resp.\ to $T^I$).
The interpretations of types are disjoint and finite (except $\mathbb{Q}^I$ which is infinite).
The interpretation $p^I$ of a predicate $p$ is a set of tuples $\bar d$ of domain elements of appropriate types.
The interpretation $f^I$ of a function $f$ is a set of pairs $(\bar d,d)$, also denoted by $\bard \mapsto d$.
%The \emph{size} of a structure is the number of tuples in the interpretation of predicate and functions.
%We only consider structures with a finite size.

An \emph{extended structure} is a structure with a variable assignment, i.e., with a mapping from variable $\bar x$ to values $\bar d$, denoted $[\bar x: \bar d]$.
The value of a formula or term in an extended structure is computed according to the usual FO semantic rules, which require that the interpretation of function symbols be total functions over their domain.
A \V-structure is a \emph{model} of a sentence if the value of the sentence is $\Tr$ in the structure, i.e., if it \emph{satisfies} it.
 \emph{Satisfiability checking} in typed FO logic is the problem of deciding whether a sentence has a model, given the interpretation of the types.

\subsection{Permutations and orbits}

A \emph{permutation} of a set is a bijection from that set to itself.
We denote a permutation by $\pi$, and its inverse by $\pi^{-1}$.
The \emph{identity} permutation, $\pi^0$, maps every element of the set to itself.
The \emph{order} of a permutation is the smallest positive number $n$ such that $\pi^n$ is the identity permutation.
Hence, the $n$th permutation of an element is equal to its ``$n$th modulo the order'' permutation.

\emph{Cycles} are permutations that map elements in a cyclic fashion. %, e.g., it maps element $a$ to $b$, $b$ to $c$, and $c$ to $a$.
A cycle is denoted by $(d^1d^2\cdots d^n)$. %, to express that $\pi(d^i)=d^{i+1}$$ for $0\leq i<n$, and $$\pi(d_n)=d^0$.
%The \emph{length} of a cycle is the size of its orbit.
A permutation over a finite set has a unique decomposition into \emph{disjoint cycles}.
Its order is the least common multiple ($\lcm$) of the length of its cycles.

The permutation of a tuple of elements $\bar d$ is denoted by $\pi(\bard)$ and is equal to $(\tuple{\pi(d_1)}{\pi(d_n)})$.
The \emph{$\pi$-orbit} of an element $d$ (resp. of a tuple of elements $\bar d$) is the set of its repeated permutations $o_{\pi}(d)=\{\pi^i(d) \mid i \in \mathbb{N}\}$ (resp. $o_{\pi}(\bar d)=\{\pi^i(\bar d) \mid i \in \mathbb{N}\}$).

%We denote by $o(d)$ the orbit containing element $d$.
    %: $\overline{(abc)}=\{a,b,c\}$.
%Hence, we can use cycle notation to compactly represent permutations.
% $(abc)(de)(f)$ is a permutation that has 3 cycles:
% the first one maps element $a$ to $b$, $b$ to $c$, $c$ to $a$; the second one swaps $d$ and $e$; and the third one maps $f$ to itself.
% one that maps element $a^J()$ to $b^J()$, $b^J()$ to $c^J()$, $c^J()$ to $a^J()$, one that swaps $d^J()$ and $e^J()$, and one that maps $f^J()$ to itself.
% https://math.stackexchange.com/questions/1419937/prove-that-the-order-of-an-element-in-s-n-equals-the-least-common-multiple-of
% The length of a cycle is the number of elements in it. %: $|(abc)|=3$.
% Cycles of length one are often omitted in the cycle notation; the elements they contain are called fixed points.
% The orbits of a permutation partition its domain.

%%%%%%%%%%%%%%%%%%%%%%%%%%%%%%%%%%%%%%%%%%%%%%%%%%%%%%%%%%%%%%%

\section{Lossless compression of structures} % and expansion
\label{sec:lifted-model}

% Recall that, in a so-called \emph{concrete} structure, each domain element in a model represents a distinct component of a solution of the configuration problem, while, in our novel \emph{lifted} structure, each domain element represents several components that are interchangeable.
% We represent this interchangeability by a symmetry.
In this section, we discuss how, and when, a concrete structure with symmetries can be compressed into a lifted structure over a smaller domain without loss of information, and how a lifted structure can be expanded into a concrete one.

Symmetries in a structure $I$ are described by a domain permutation $\pi$~\cite{DBLP:journals/tplp/DevriendtBBD16}:
it is a permutation of the domain of $I$ that maps numbers to themselves, and other domain elements in $T^I$ to domain elements in $T^I$.
Since numbers are mapped to themselves by the permutation, and since types besides $\mathbb{Q}$ are finite, the permutation is composed of cycles.

A domain permutation \emph{induces a structure transformation}.
The interpretation in the transformed structure, denoted $\pi(I)$, is defined as follows:
\begin{itemize}
    \item the type of domain element $d$ in $\pi(I)$ is its type in $I$;
    \item a tuple $\bard'$ is in the transformed interpretation of predicate symbol $p$ if and only if $\pi^{-1}(\bard')$ is in the original interpretation of $p$: \\
    $ p^{\pi(I)} = \{\bard' \mid  \pi^{-1}(\bard') \in p^I\} = \{\pi(\bard) \mid  \bard \in p^I\}$;
    \item a tuple $(\bard' \mapsto d')$ is in the transformed interpretation of function symbol $f$ if and only if $\pi^{-1}(\bard')\mapsto  \pi^{-1}(d')$ is in the original interpretation of $f$:
    \begin{align*}
        f^{\pi(I)} &= \{(\bard' \mapsto d') \mid  (\pi^{-1}(\bard') \mapsto \pi^{-1}(d')) \in f^I\} \\
        &= \{(\pi(\bard) \mapsto \pi(d)) \mid  (\bard \mapsto d) \in f^I\}
    \end{align*}
\end{itemize}

\begin{definition}[Automorphism]
    A permutation $\pi$ that transforms a structure into itself, i.e., such that $I = \pi(I)$,  is an \emph{automorphism} of the structure.
\end{definition}

Every structure has at least one automorphism: the identity domain permutation.
Note that an automorphism maps the interpretation of any constant to itself, i.e., $c^I() = \pi(c^I())$.
Thus, the length of the cycle containing $c^I()$ is 1.

We introduce the new concept of \emph{backbone}, which plays a critical role in the compression that we propose.
Essentially, a \emph{backbone} for an automorphism of $I$ is a set of domain elements obtained by picking one element in each cycle, such that the interpretation of any symbol can be reconstructed from its interpretation restricted to the backbone, by applying the automorphism repeatedly.
Formally:

\begin{definition}[Backbone]
\label{def:backbone}
    A \emph{backbone} for an automorphism $\pi$ of $I$ is a subset $S$  of the domain of $I$ such that:
    \begin{itemize}
        \item each cycle $C$ of $\pi$ has exactly one element in $S$;
        \item for each predicate $p/n\in \Sigma$, $p^I$ is the union of the $\pi$-orbits  of the tuples in $p^\I \cap S^n$, i.e.,
        \begin{align}
            p^I = \bigcup_{\bard \in p^\I \cap S^n} o_{\pi}(\bar d) = \{\pi^i(\bar d)\mid \bard \in p^\I \cap S^n, i \in \mathbb{N}\}
        \end{align}
        \item for each function $f/n\in\Sigma$, $f^I$ is the union of the $\pi$-orbits of the tuples in $f^I\cap (S^n\times S)$, i.e,
        \begin{align}
            &f^I = \bigcup_{(\bard \mapsto d) \in f^\I \cap (S^n \times S)} o_{\pi}(\bard \mapsto d)  \\
                &= \{\pi^i(\bard \mapsto d)\mid (\bard \mapsto d) \in f^\I \cap (S^n \times S), i \in \mathbb{N}\} \nonumber
        \end{align}
    \end{itemize}
\end{definition}

\begin{example}
    For the pigeonhole example in the introduction of the paper, a backbone of automorphism $(p_1\cdots p_{10})(h_1\cdots h_5)$ is $S=\{p_1, h_1\}$.  Another is $S=\{p_2, h_2\}$.
\end{example}

In all structures, the set of domain elements is a \emph{trivial} backbone for the identity automorphism. However, not all automorphisms  have a backbone.

\begin{example}\label{ex:iso}
    Let $\{a,b\}$ be a (concrete) domain with one type $T$.
    In structure $I_1$ (resp.\ $I_2$), function symbol $f: T \rightarrow T$ is interpreted as $\{a \mapsto a, b \mapsto b\}$ (resp. $\{a \mapsto b, b \mapsto a\}$).
    Permutation $(a b)$ is an automorphism of both structures.
    The only two subsets $S$ satisfying the first condition of backbone for $(ab)$ are $\{a\}$ and $\{b\}$. Since $f^{I_1}$ can be reconstructed from $f^{I_1}\cap (S\times S)$ for both subsets $S$, both subsets are backbones of $I_1$. However, none is a backbone of $I_2$ (because $f^{I_2}\cap (S\times S)=\emptyset$ for both candidate sets $S$).
\end{example}

A backbone enables us to lift a structure into a structure with a smaller domain, as we now describe.

\begin{definition}[Lifted vocabulary]
    For a vocabulary $\Sigma$, its lifted vocabulary $\Sigma_l$ consists of the symbols of $\Sigma$ and, for any type $T$ of $\Sigma$, a symbol $\mul_T:T \to \mathbb{N} $, called the  \emph{multiplicity} function for $T$.

\end{definition}

We will drop the subscript $T$ in the function symbol $\mul_T$ when this is unambiguous.
%Note that the multiplicity of any constant, including numerals, is one.

\begin{definition}[Lifted structure]
    \label{def:lifted-structure}

    Let $I$ be a $\Sigma$-structure with an automorphism $\pi$ having  backbone $S$.
    A lifted structure $L$  derived from $I$ is a $\Sigma_l$-structure  such that:
    \begin{itemize}
    \item its domain is $S$, called the lifted domain;
    \item for each type predicate $T$, $T^L=T^I \cap S$;
    \item for each predicate symbol $p/n$, $p^L=p^\I\cap S^n$;
    \item for each function symbol $f/n$, $f^L=f^I \cap (S^n\times S)$;
    \item for each $l\in S$, $mul^L(l)=|o_{\pi}(l)|$, the size of the $\pi$-orbit of $l$ in $I$.
    \end{itemize}
\end{definition}

\begin{example}
    Continuing the pigeonhole problem, the lifted structure is described by Equations~\ref{ph:l1}-\ref{ph:l4}.
\end{example}

% \begin{example}
%     Continuing the previous example, the lifted structure of $I_1$ for backbone $\{a\}$ of automorphism $(ab)$ is the structure $L_1$ over the domain $\{a\}$ with $f$ interpreted by $\{a \mapsto a\}$ and $\mul_T(a)=2$.
%     The lifted structure of $I_2$ for the backbone $\{a,b\}$ of the identity automorphism $(a)(b)$ is the structure over the domain $\{a,b\}$ in which the multiplicity of $a$ and $b$ is one, and $f$ is interpreted by $\{a \mapsto b, b \mapsto a\}$.
% \end{example}

% \begin{proposition}
%     For the identity automorphism of a structure $I$, the set $S=dom(I)$ is a backbone and the lifted structure $L$ is identical to $I$ (ignoring the multiplicity symbols).
% \end{proposition}

Given the lifted structure $L$ derived from $I$ and the automorphism $\pi$ on the concrete domain $I$, one can reconstruct $I$, i.e., one can \emph{expand} the lifted interpretations of the type, predicate and function symbols, essentially by closing them under repeated application of $\pi$.
Formally,
\begin{align}
    \Exp(T^L) &= \bigcup_{l\in T^L} o_{\pi}(l) = \{ \pi^i(l) \mid  l\in T^L , i\in \mathbb{N}\} \label{def:expansion-type}\\ % \bigcup_{i\in N}\pi^i(T^L)=
    \Exp(p^L) &= \bigcup_{\bar l\in p^L} o_{\pi}(\bar l) =\{ \pi^i(\barl) \mid  \barl\in p^L , i\in \mathbb{N}\} \label{def:expansion-predicate}\\ % \bigcup_{i\in N}\pi^i(p^L)=
    \Exp(f^L) &= \bigcup_{(\bar l \mapsto l)\in f^L} o_{\pi}(\bar l \mapsto l) \nonumber\\
        &= \{ \pi^i(\barl \mapsto l) \mid  (\barl \mapsto l)\in f^L , i\in \mathbb{N}\}. \label{def:expansion-function} % \bigcup_{i\in N}\pi^i(f^L)=
\end{align}
%\bbcomment{Why $i\in \mathbb{N}$ and not $i\in \mathbb{N}$?}

\begin{example}
    Continuing the pigeonhole problem, the expansion of the lifted structure described by Equations~\ref{ph:l1}-\ref{ph:l4} for automorphism $(p_1\cdots p_{10})(h_1\cdots h_5)$ is the structure described by Equations~\ref{ph:c1}-\ref{ph:c3}.
\end{example}

In our approach, we need to find lifted structures that can be expanded into concrete ones.
We observe that there is a simple construction to expand \emph{any} lifted $\Sigma_l$-structure $L$ with lifted domain $S$ that results sometimes (but not always) in a concrete structure $I$ with an automorphism $\pi$ having backbone $S$.
This construction is as follows.
%We construct $S_0$ (resp. $L_0$) by restricting $S$ (resp. $L$) to lifted elements with strictly-positive multiplicities.

\begin{definition}[Expansion of lifted domain]
    An \emph{expansion of the domain} $S = \{l_i\mid i \in \mathbb{N}\}$ of the lifted $\Sigma_l$-structure $L$ is a set $D=\{d_i^j \mid \exists i,j \in \mathbb{N}:l_i \in S \land 1\leq j\leq \Mul^L(l_i)\}$ such that the $d_i^j$ are distinct and $\forall i: d_i^1=l_i$.
\end{definition}

%We construct $S_0$ (resp. $L_0$) by restricting $S$ (resp. $L$) to lifted elements with strictly-positive multiplicities.
We call $d_i^1$ a \emph{base} element.
Notice that  lifted domain elements with multiplicity zero have no image in D.\footnote{Null multiplicities will prove useful for iterative methods in Section~\ref{sec:evaluation}}
We define a permutation on the set $D$ having the cycles $(l_i, d_i^2\dots d_i^{\Mul^L(l_i)})$:

\begin{definition}[Permutation of the expanded domain]
    The permutation of an expansion $D$ of a lifted domain is the function $\pi: D \mapsto D$ such that $\pi(d_i^j)=d_i^{j+1}$ for $1 \leq j<\Mul^L(l_i)-1$, and $\pi(d_i^{\Mul^L(l_i)})=d_i^1$.
\end{definition}

The expansion of a lifted tuple $\bar l$ in a lifted structure is its $\pi$-orbit:
\begin{equation} \label{exp:tuple}
    \Exp(\bar l) = o_{\pi}(\bar l) = \{\pi^i(\bar l) | i \in \mathbb{N}\}
\end{equation}
if none of the expansion of its elements is empty, and is empty otherwise.
The concrete interpretation of symbols is obtained by expanding the tuples in their lifted interpretations, as in Equations (\ref{def:expansion-type}-\ref{def:expansion-function}).

    % Observe that the lifted domain is smaller than the original domain, and that the lifting can be reverted: given the lifted structure of an original structure, one can derive a new structure isomorphic to the original structure, by i) determining the interpretation restricted to the backbone of the new structure (using Definition~\ref{def:lifted-structure}), and ii), extending it to the whole domain by repeated application of the permutation (using Definition~\ref{def:backbone}).
    % So, lifting is a form of \emph{lossless compression} of the original structure.
    % The reverse operation is called \emph{expansion}.
    % We now explain how this expansion can be performed.

It follows from the correspondence between Definition~\ref{def:backbone} and Equations (\ref{def:expansion-type}-\ref{def:expansion-function}) that, if $\pi$ and $I$ are constructed from $L$ in this way, and if $I$ is actually a structure, then $S$ is a backbone of $\pi$ in $I$.
When $L$ only has strictly positive multiplicities, a lifted structure derived from the expansion of $L$ will be isomorphic to $L$: in some sense, the compression is lossless.
% The lifted structure derived from the expansion $I$ of $L$ has only strictly positive multiplicities, and will be isomorphic to  $L$ only if $L$ only has strictly positive multiplicities.
% All good, but is $I$ not necessarily a structure?
However, $I$ may not be a structure because the interpretation of a function symbol in $I$ might not be a total function on its domain, as the following example shows.

\begin{example}
    Consider the lifted structure $L$ with  domain $S = \{a, b\}$, with types $A^L=\{a\}$, $B^L=\{b\}$, and for $f: A\times A \mapsto B$,  $f^L= \{(a, a) \mapsto b\}$. Finally,  let $\mul(a)=\mul(b)=2$,
    The expanded domain is $\{a^1, a^2, b^1, b^2\}$ with permutation $(a^1a^2)(b^1b^2)$.
    The expansion of $f^L$ is $\{(a^1, a^1) \mapsto b^1, (a^2, a^2) \mapsto b^2\}$.
    There is no entry for  $(a^1, a^2)$, so the expansion is not a total function.

    Consider now the lifted structure $L$ with  domain $S = \{a, b\}$, with type $A^L=\{a\}, B^L=\{b\}$, for $f: A\mapsto B$, $f^L= \{a \mapsto b\}$. Finally,  $\mul(a)=1, \mul(b)=2$.
    The expanded domain is $\{a^1, b^1, b^2\}$ with permutation $(a^1)(b^1b^2)$. The expansion of $f^L$ is $\{a^1 \mapsto b^1, a^1 \mapsto b^2\}$.
    This expansion gives two different values for  $a^1$, so it is not  a function.
\end{example}

To distinguish lifted structures that can be expanded into concrete ones from those that cannot, we introduce the notion of \emph{regularity}.

\begin{definition}[Regular function, regular lifted structure]
    \label{def:reg-structure}\label{def:reg-struct}
    The lifted interpretation of a function symbol is regular if its expansion defines a total function in the concrete domain, i.e., if it specifies exactly one image for every tuple in the concrete domain.
    A \emph{regular lifted structure} is a lifted structure in which the interpretation of each function symbol is regular.
\end{definition}

Now, we define the expansion of a regular lifted structure:

\begin{definition}[Expansion of a regular lifted structure]
    Let $L$ be a regular lifted structure over a lifted vocabulary.
    Then, $I$, the expansion of $L$, is the structure over the concrete vocabulary defined as follows:
    \begin{itemize}
        \item its domain is the expansion of the lifted domain, having permutation $\pi$ derived from $L$; %union of each type expansion, i.e., $\bigcup_{T_i} \Exp(T_i^L)$;
        \item for each type symbol $T$, $T^I = \Exp(T^L)$;
        \item for each predicate symbol $p$, $p^I =\Exp(p^L)$;
        \item for each function symbol $f$, $f^I =\Exp(f^L)$.
    \end{itemize}
\end{definition}

The expansion of a regular lifted structure does not involve any search.
% If $O$ is the order of the permutation induced by $L$, the expansion of a tuple has at most $O$ tuples.
% Hence, the expansion of $L$ is at most $O$ times larger than the size of $L$.
% If $M$ is the maximal multiplicity of a lifted element and $N$ the largest arity of symbols in $\Sigma$, the expansion of a tuple in a symbol interpretation has at most $M^N$ tuples, and the expansion of $L$ is at most $M^N$ times larger than $L$.
The time needed for this expansion is generally negligible (e.g., less than 0.1 sec for 10,000 concrete domain elements).

    % Another natural property is that the operations of computing the lifted structure, and re-expanding it commute. In particular, if $L$ is the lifted structure of a structure $I$ using automorphism $\pi$ with backbone $S$, then the expansion $I'$   of $L$ is isomorphic to $I$. In particular, the isomorphism is the mapping of the lifted elements $d_i$ to $\pi^i(d)$.  PROOF IN APPENDIX : EASY
To further characterize regular functions, we introduce the concept of regular tuples.
We use $\Mul(l_1,\ldots,l_n)$ (resp. $\LCM(l_1,\ldots,l_n)$) as a shorthand for the product of the multiplicities of $l_i$ (resp. their least common multiple $\lcm$).
First, we observe that the size of the expansion (Equation~\ref{exp:tuple}) of a tuple $\bar l$ is finite: it is the order of the permutation defined by the cycles of its elements, i.e., $\LCM(l_1, \ldots , l_n)$.\footnote{Taking the convention that the $\lcm$ of a tuple of numbers containing 0 is 0.}
Also, the expansion has at most $\Mul(l_1,\ldots,l_n)$ tuples; it is then the cross-product of the expansions of its elements.
When those two numbers are identical, we say that the tuple is regular.

\begin{definition}[Regular lifted tuple]
    \label{def:regular-tuple}
    A lifted tuple is regular if and only if its expansion is the cross-product of the expansion of its elements.
\end{definition}

\begin{example}
  Let $(a,b)$ be a tuple of two lifted domain elements with $\mul(a)=2$ and $\mul(b)=4$.
  Its expansion is $\{(a^1,b^1),(a^2,b^2),(a^1,b^3),(a^2,b^4)\}$, of size $\LCM(2,4)=4$.
  Note that, e.g., the tuple $(a^1,b^2)$ does not belong to the expansion: thus, $(a,b)$ is not regular.
  It is regular when $\mul(b)=0$ (the expansion of $b$ and of $(a,b)$ are empty), or when, e.g., $\mul(b)=3$ (the expansion is $\{(a^1,b^1),(a^2,b^2),(a^1,b^3),(a^2,b^1),(a^1,b^2),(a^2,b^3)\}$), of size $\LCM(2,3)=6$.
\end{example}

%The expansion of two distinct tuples are disjoint (because the expansion of lifted elements are disjoint).
Nullary and unary tuples are always regular.
An $n$-ary tuple is regular when one of its elements has multiplicity zero, or every pair of its elements have multiplicities that are coprime.

\begin{proposition}[Regular function] %[Regularity condition for function]
    \label{def:reg-function}
    A function $f^L$ is regular if, for all tuples $\bar l$ in the domain of $f^L$, it holds that
    (i) $\bar l$ is regular, and
    (ii) the multiplicity of $\bar l$ is a multiple of the multiplicity of its image.
\end{proposition}

\begin{proof}

    First, we show that
    %The first regularity condition (i) ensures that
    the expansion of $f^L$ gives at least one value for every tuple $\bard$ in the concrete domain of $f$.
    %Indeed, e
    Each element $d_i$ of $\bar d$ is in the expansion of a lifted domain element $l_i \in T_i^L$.
    Tuple $\bar l= (l_1, \cdots, l_n) \in \bar T^L$ is in the lifted domain of $f^L$ as $f^L$ is total; it is regular by (i), hence $\bard $ is in its expansion and $f^I(\bard)$ is in $\Exp(f^L(\barl))$.

    %Then, tuple $\bar l= (l_1, \cdots, l_n) \in \bar T^L$; it is in the lifted domain of $f^L$, and is regular by (i): its expansion is the cross-product of the expansion of its elements.
    %Hence, tuple $\bar d$ is in the expansion of $\barl$, and $f^I(\bard)$ is in $\Exp(f^L(\barl))$.

   Next, we show that the expansion
   %Condition (ii) ensures that it
   gives at most one value for every tuple $\bard$ in the concrete domain of $f$.
   %Indeed, l
   Let $\bard$ be in the expansion of $\barl$, with $\LCM(\bar l)=m=\MUL(\bar l)$ by (i).
   We thus have $0 < m$.
   The expansion of $f^L$ contains the pairs $\pi^i(\barl \mapsto f^L(\barl))$ and $\pi^{i+n\times m}(\barl \mapsto f^L(\barl))$, for any $\barl$ in the domain of $f^L$, and for any $i$ and $n$.
   The first element of these two pairs are identical by definition of $m$; the second elements are identical by (ii).
\end{proof}

     \section{Translation into a lifted sentence}
    \label{sec:lifted-theory}

%\bbcomment{equisatisfiable for a given size, this means something like ``fix the domain (size)'', one has a model of this size iff the other has a model that corresponds (taking mul into account) to this size? Maybe make this explicit}
% could not find a concise way to do it yet.
%The translation of a sentence $\F$ consists of (i) the transformed sentence $\tr(\F)$, (ii) sentences expressing regularity conditions, and (iii) sentences stating the size of the concrete types (i.e., of the sum of values of $\mul_T$).
    We now present a translation of a sentence in \fota~into a sentence that allows domain compression, such that the translation is equisatisfiable with the original. The translation $\overline{\chi(\phi)}$ of a sentence $\F$ is the conjunction of (i) the transformed sentence $\tr(\F)$, and (ii) sentences expressing regularity conditions.

    The transformation $\chi(e)$ of an expression $e$ is defined recursively in Table~\ref{table:trans}.
    The left column shows the possible syntactical forms in the concrete sentence;
    the middle column shows the transformation;
    the third column shows the regularity constraints added to the translation.
    The bottom part of the table shows the regularity constraints added for each function symbol in the vocabulary to ensure that the function interpretations are regular (Proposition~\ref{def:reg-function}).

    \begin{table*}[t]
    \centering
    \caption{Terms, formulas, and their translation}
    \label{table:trans}\begin{threeparttable}
        \begin{tabularx}\textwidth{@{}MMML@{}}
            \text{Expression $e$}               & \text{Transformation $\chi(e)$} & \text{Regularity condition} \\
            \toprule
            \multicolumn{1}{c}{\text{Term}}                & ~& ~  \\
            \cmidrule{1-1}
            x  &  x  & ~& eq:Xx\\
            c()                                 & c() & ~&  eq:Xc \\
            f(\bart)     & f(\ch{\bart}) & ~   & eq:Xf \\
            t_1 \oplus t_2, \oplus \in \{+, -, \times, \div\}     & \chi(t_1) \oplus \chi(t_2) & ~& eq:Xop \\
            \SUM{\barx \in \barT}{ p(\bar t, s) \land \phi} {t} \tnote{(*)}  & \multicolumn{2}{l}{\SUMW{\barx \in \barT}{ 0<\MUL(\barx) \land p(\chi({\bar t}),\chi(s))   \land \chi(\phi)}{\frac{\Mul(\bar x) \LCM(\tr(\bar t),\tr(s))}{\LCM(\tr(\bar t)) \mul(\tr(s))}\times \chi(t)} } & eq:Xsumf2 \\
            \SUM{\barx \in \barT}{\phi} {t}     & \multicolumn{2}{l}{\SUMW{\barx \in \barT}{0<\MUL(\barx) \land\ch{\phi}} {\MUL(\barx) \times \ch{t}} } & eq:Xsum \\

            \midrule%%%%%%%%%%%%%%%%%%%%%%%%%%%%%%%%%%%%%%%%%%%%%%%%%%
            \text{Formula}               & ~  & ~\\
            \cmidrule{1-1}

            p()  & p() & ~& eq:Xprop\\
            p(\bart)      & p(\ch{\bart}) & \RC(p(\ch{\bart})) \tnote{(**)} & eq:Xp\\
            t_1 \sim t_2, {\sim} \in \{<,>,\leq,\geq \}   & \ch{t_1} \sim \ch{t_2} & ~ & eq:Xeqa \\
            %
            %t_1 = t_2 & \ch{t_1} = \ch{t_2} & \RC(\ch{t_1} = \ch{t_2}) \tnote{(***)} & eq:Xeqb \\

            %
            \phi_1 \bigotimes \phi_2, \bigotimes \in \{\land, \lor, \Rightarrow, \Leftrightarrow \}                   & \ch{\phi_1} \bigotimes \ch{\phi_2} & ~& eq:Xand\\
            \lnot \phi  & \lnot \ch{\phi} &  ~ & eq:XNot\\
            T \triangleq \{t_1, \ldots, t_n\} & \sum_{x \in T} {\mul(x)} = n &~ & eq:T \\
            \forall \barx \in \barT: p(\bar t,s) \Rightarrow \phi \tnote{(*)} &  \multicolumn{2}{l}{$\forall \barx \in \barT : 0 < \MUL(\barx) \land p(\chi({\bar t}),\chi(s)) \Rightarrow \chi(\phi) $} & eq:Xforall2 \\
            \exists \barx \in \barT: p(\bar t,s) \land \phi \tnote{(*)} &  \multicolumn{2}{l}{$\exists \barx \in \barT : 0 < \MUL(\barx) \land  p(\chi({\bar t}),\chi(s)) \land \chi(\phi) $} & eq:Xexists2 \\
            \forall \barx \in \barT: \phi &  \multicolumn{2}{l}{$\forall \barx \in \barT : 0 < \MUL(\barx) \Rightarrow \chi(\phi) $} & eq:Xforall \\
            \exists \barx \in \barT: \phi &  \multicolumn{2}{l}{$\exists \barx \in \barT : 0 < \MUL(\barx) \land \chi(\phi) $} & eq:Xexists \\
        \midrule
            %\text{for each function $f: \barT \rightarrow T$} & ~ & ~  \\
            \cmidrule{1-1}
            \text{for each function $f: \barT \rightarrow T$}& \multicolumn{2}{r}{$\forall \barx \in \barT: \Mul(\barx)=\LCM(\barx)$} & eq:f1 \\
            \cmidrule{1-1}
            ~& \multicolumn{2}{r}{$\forall \barx \in \barT: \exists n \in \mathbb{N}: \Mul(\barx)=n \times \mul(f(\barx))$} & eq:f2 \\

        \bottomrule
        \end{tabularx}
        \begin{tablenotes}[flushleft]
            \footnotesize
            \item[(*)] Rules~(\ref{eq:Xsumf2}, \ref{eq:Xforall2}, \ref{eq:Xexists2}) are applied only when $\vars(\bar t) \subseteq \{\barx\}$ and $\vars(s) \cap \{\barx\} = \emptyset$; Rules~(\ref{eq:Xsum}, \ref{eq:Xforall}, \ref{eq:Xexists}) are applied otherwise.
            \item[(**)] $\RC(p(\tr(\bar t)))$ is defined as $\forall \bar x \in \bar T_{\barx}: p(\tr(\bar t)) \Rightarrow \Mul(\barx)=\LCM(\barx) \lor \Mul(\tr(\bar t))= \LCM(\tr(\bar t))$, where $\barx=\free(p(\tr(\bar t)))$.
            %\item[(***)] %This regularity condition is added only if $t_1$ and $t_2$ are not numeric.
            %This regularity condition is unnecessary and can be dropped when $t_1$ and $t_2$ are numeric.
        \end{tablenotes}
        \end{threeparttable}
    \end{table*}

    Note that sum aggregates and quantified formulas are transformed by specialized rules (Rules~\ref{eq:Xsumf2}, \ref{eq:Xforall2}, \ref{eq:Xexists2}) when possible (see table footnote), and by general rules (Rules~\ref{eq:Xsum}, \ref{eq:Xforall}, \ref{eq:Xexists}) otherwise.
    %Predicate $p$ appearing in the special rules can be a pre-defined predicate (i.e., equality and arithmetic comparison).
    Generally, it is beneficial to do equivalence-preserving transformations of sentences to obtain sentences of the form allowing application of specialized rules. %~(\ref{eq:Xsumf2}, \ref{eq:Xforall2}, \ref{eq:Xexists2}).
    The specialized rules do not require the translation of $p(\bar t,s)$ by Rule \ref{eq:Xp}, thus avoiding the regularity constraint $RC(p(\chi(\bar t, s)))$.
    E.g., for the atom $t=s$, this regularity constraint would enforce $\mul(t)=mul(s)=1$ for each tuple $(t,s)$ in the lifted equality relation, significantly reducing the possibility of compression.

    The transformation of a sentence consists of adding $0 < \Mul(\bar x)$ filters (to cope with lifted domains elements with an empty expansion) and of multiplying the term $t$ in an aggregate term with some decompression factor:
    in Rule~\ref{eq:Xsum}, it is the number $\Mul(\bar x)$ of possible concrete variable assignments;
    in Rule~\ref{eq:Xsumf2}, it is that number multiplied by the fraction of concrete assignments that make $p(\bar t, s)$ true.
    % (when a formula is true in the lifted structure for some variable assignment of its free variables, it can be true in the expanded structure of $L$ for \emph{several} variable assignments in the expansion of the lifted one).
    The regularity condition added for an atom (Rule~\ref{eq:Xp}) ensures the translated atom is equisatisfiable with the original atom.

    % Pierer: I don't think that the following is useful.  In a satisfiability problem, the domain is known, and elements have an identifier.  In the lifted domain, one has to add a constraint on size of each type.
    % It is also beneficial to avoid constants in the original sentence.
    % For example, when asserting that the number of pigeons is 100, one lifted domain element with multiplicity 100 can be enough; if the pigeons are identified by the numerals 1 to 100 in the original sentence, 100 lifted domain elements with multiplicity 1 are needed and no compression is possible for the pigeon domain.\todo{This discussion needs rework.}
    % It is also beneficial to avoid identifying domain elements by constants in the original sentence.
    % For example, it is better to assert that the number of pigeons is 100 (using a cardinality constraint) than to assert that the pigeons are identified by the numerals 1 to 100.

    \begin{example}
        \label{ex:pigeon}
        The sentence "at most 2 pigeons in each hole": \\
        \indent $\forall h \in \Hole: \#\{p \in \Pigeon \mid \isIn(p,h)\}\leq 2$. \\
        Its transformation by Rule~\ref{eq:Xsumf2} (with $\phi=\Tr$) and \ref{eq:Xforall} is:\footnote{Recall that a cardinality aggregate is a shorthand for a sum aggregate whose term is $1$.}\\
        \indent    $\forall h \in \Hole: 0 < \mul(h) \Rightarrow \\
        \indent\indent \SUMW{p \in \Pigeon }{ 0 < \mul(p)  \land \isIn(p,h) } {\frac{\LCM(p,h) }{\mul(h)}} \leq 2.$

    \end{example}

\begin{theorem}[Equisatisfiability]
    An \fota~sentence is equisatisfiable with its translation, given the number of concrete elements in each type.
    %If a structure $I$ satisfies a sentence, then there is a lifted structure that satisfies the lifted sentence and whose expansion is $I$.

\end{theorem}
If $I$ is a model of the sentence, then $L$ constructed by extending $I$ by setting all multiplicities to one is a model of the translated sentence: indeed, the added constraints are trivially satisfied, and the translated sentence is equivalent to the original one. Proving the converse
%, i.e., that if $L$ is a model of the translated sentence, then its expansion is a model of the original sentence,
is long and complex. It is proved by structural induction of two invariants over the syntactic tree of the sentence to be satisfied:
(i) a transformed formula $\chi(\F)$ is true in $L$ under some variable assignment $[\bar x: \bar l_x]$ if and only if $\F$ is true in $I$ under any variable assignments $[\bar x: \bar d_x]$ such that each $d_{x_{i}}$ is in the expansion of $l_{x_{i}}$
(ii) similarly, if a transformed term $\chi(t)$ has value $l$ in the lifted structure $L$ under some variable assignment $[\bar x: \bar l_x]$, then, the expansion of the value $l$ contains the value of the term $t$ in $I$ (the expansion of $L$), for any variable assignment $[\bar x: \bar d_x]$ such that each $d_{x_{i}}$ is in the expansion of $l_{x_{i}}$.
%In particular, aggregate terms have the same value in $I$ and $L$ (because numbers have multiplicity 1).
This property holds %only for terms and sub-formulas occurring in the sentence, and
only when the regularity constraints given in Table~\ref{table:trans} hold in the lifted structure.
This explains why these constraints are added to the transformed sentence.
The proof is in Appendix A.

%  END OF IGNORE OF THE SECTION

%(Rules~\ref{eq:Xsumf2}, \ref{eq:Xforall2}-\ref{eq:Xexists2})

\section{Evaluation of the method}
\label{sec:evaluation}

%We describe our implementation, validate it on simple pigeonhole problems, and evaluate it on generative configuration problems.  We discuss applications on BAPA problems.

\paragraph{Implementation}

The goal of the evaluation is to show that there are satisfiability problems where substantial compression of the domain is possible, and that the lifted models can indeed be expanded into concrete ones.
A problem is solved iteratively, starting with empty lifted domains. Given a domain, the lifted sentence is reduced to a propositional sentence and its satisfiability is determined with a standard satisfiability solver capable of arithmetic reasoning.
If the sentence is unsatisfiable with this domain, the sentence is reduced to a minimal unsatisfiable formula~\cite{DBLP:conf/sat/LynceM04}, and the domains of the types used in that formula are extended with one element.
This process is repeated until a model of the sentence is found (it  does not terminate if the original sentence is unsatisfiable for any domain size, unless one imposes an upper limit on the size of lifted domains).

In many experiments, it was sufficient to support the special Rules~\ref{eq:Xsumf2}, \ref{eq:Xforall2}-\ref{eq:Xexists2} for the case where the atom $p$ is of the form $t_1=s$ (e.g., $\holeOf(p)=h)$.
Then, the transformation of a sum aggregate per Rule~\ref{eq:Xsumf2} simplifies to:
\begin{align}
    \frac{1}{\mul(\chi(s))} \SUMW{\barx \in \barT}{  0<\MUL(\barx) \land \chi({t_1}){=}\chi(s)   \land \chi(\phi)  \nonumber}{& \MUL(\bar x)\times \chi(t) \\ &}
\end{align}
This formula does not use the $\lcm$ function, which is not supported natively by solvers.
In other experiments, we used an interpretation table for $\lcm$.
%Our prototype implementation only supports this form, which avoids the use of the $lcm$ function.
%The $lcm$ function is not supported natively by solvers, and finding an efficient implementation has been left for future work.
%We have to translate problems manually when they need precision beyond the initial prototype.

%The iterative method may introduce more lifted domain elements than needed for a solution, in which case unneeded domain elements have multiplicity 0 (this introduces symmetry in the solution of a lifted theory; it can be reduced by ordering the lifted domain elements $l_0,l_1,\ldots $ of a type and adding the constraint $\mul(l_0 \geq \mul(l_1) \geq \mul(l_2) \ldots$).

Problems are expressed in $\mathrm{FO}(\cdot)$\footnote{https://fo-dot.readthedocs.io/}, a Knowledge Representation language with support for types, subtypes, and aggregates.\footnote{Subtypes are subsets of types, and are declared as unary predicates in $\mathrm{FO}(\cdot)$.  Subtypes can be used where types are used: in the type signature of a symbol and in quantification.}
A $\mathrm{FO}(\cdot)$ (lifted) sentence is translated by IDP-Z3~\cite{DBLP:journals/corr/abs-2202-00343} for use with the Z3 SMT solver~\cite{de2008z3}.
The overhead of this translation is negligible.

The source code of our examples is available on GitLab\footnote{\label{repository}\url{https://gitlab.com/pierre.carbonnelle/idp-z3-generative}}.
Tests were run using Z3 v4.12.1 on an Intel Core i7-8850H CPU at 2.60GHz with 12 cores, running Ubuntu 22.04 with 16 GB RAM.
We run a modified version of IDP-Z3 v0.10.8 on Python 3.10.

%\subsection{Pigeonhole problems}\label{sec:pighole}
\paragraph{Pigeonhole problem}
To validate the approach, we first consider the satisfiability problem of assigning each pigeon to one pigeonhole, such that each pigeonhole holds (at most) 2 pigeons.
Function $\holeOf: \Pigeon \mapsto \Hole$ is used to represent this relation.

When there are twice as many pigeons as holes, the lifted solution has 1 lifted pigeon and 1 lifted hole, as shown in the introduction of the paper.  As expected, the time needed to solve the lifted problem is almost constant when it is satisfiable. The correct multiplicities are found quickly by Z3 using a sub-solver specialized for arithmetic.
For example, with 10,000 pigeons, the lifted sentence is solved in only 0.05 sec and the expansion of the lifted model into a concrete one in 0.1 sec. With the same solver and the original sentence, the solution time increases exponentially (4 sec to solve the problem for 30 pigeons). We are aware that symmetry breaking can reduce the complexity, but to the best of our knowledge, solving time is at least linear in the number of pigeons for the best symmetry breaking solvers.

We also validated our method on pigeonhole problems where the relation between pigeon and holes is represented by a binary predicate.  The translated sentence uses an interpretation table for $\lcm$.  Experiments confirm equisatisfiability, but the complexity is quadratic because of the $\lcm$ table.
Finding a more efficient implementation has been left for future work.

\paragraph{Generative configuration problems}
\label{sec:realistic}

Generative configuration problems (GCP) are configuration problems in which the number of some components has to be found: the number of elements in some types is not known in advance.
An iterative method is thus always required to find them.
When the compressed domain is smaller than the concrete domain, the number of iterations needed to solve the lifted sentence is smaller than the number of iterations for the original sentence, leading to better performance.
Hence, GCP is a good application domain for our method.

We evaluate our solving method on three representative GCP discussed in the literature:
\begin{itemize}
    \item the House Configuration and Reconfiguration problem~\cite{DBLP:conf/confws/FriedrichRFHSS11},
    \item the Organized Monkey Village~\cite{DBLP:conf/sat/Reger0V16}
    \item the Rack problem~\cite{DBLP:conf/models/FeinererSS11, DBLP:conf/splc/Comploi-TaupeFS22},
\end{itemize}

These problems are expressed using only an equality predicate and unary symbols, and most regularity constraints introduced by our method are trivially satisfied.

\num\newcommand{\Rack2011}{Rack 2011} % ~\cite{DBLP:conf/models/FeinererSS11}}
\num\newcommand{\RackA1}{Rack A1} % ~\cite{DBLP:conf/splc/Comploi-TaupeFS22}}
\num\newcommand{\RackA2}{Rack A2} % ~\cite{DBLP:conf/splc/Comploi-TaupeFS22}}
\num\newcommand{\RackA3}{Rack A3} % ~\cite{DBLP:conf/splc/Comploi-TaupeFS22}}
\num\newcommand{\RackA4}{Rack A4} % ~\cite{DBLP:conf/splc/Comploi-TaupeFS22}}
\num\newcommand{\RackA5}{Rack A5} % ~\cite{DBLP:conf/splc/Comploi-TaupeFS22}}
\num\newcommand{\RackA10}{Rack A10} % ~\cite{DBLP:conf/splc/Comploi-TaupeFS22}}
\num\newcommand{\RackA50}{Rack A50} % ~\cite{DBLP:conf/splc/Comploi-TaupeFS22}}
\num\newcommand{\RackA1000}{Rack A1000} % ~\cite{DBLP:conf/splc/Comploi-TaupeFS22}}

\num\newcommand{\RackABCD4}{Rack ABCD4} % ~\cite{DBLP:conf/splc/Comploi-TaupeFS22}}
\num\newcommand{\RackABCD8}{Rack ABCD8} % ~\cite{DBLP:conf/splc/Comploi-TaupeFS22}}
\num\newcommand{\RackABCD12}{Rack ABCD12} % ~\cite{DBLP:conf/splc/Comploi-TaupeFS22}}
\num\newcommand{\RackABCD16}{Rack ABCD16} % ~\cite{DBLP:conf/splc/Comploi-TaupeFS22}}
\num\newcommand{\RackABCD20}{Rack ABCD20} % ~\cite{DBLP:conf/splc/Comploi-TaupeFS22}}

\num\newcommand{\RackC1}{Rack C1} % ~\cite{DBLP:conf/splc/Comploi-TaupeFS22}}
\num\newcommand{\RackC2}{Rack C2} % ~\cite{DBLP:conf/splc/Comploi-TaupeFS22}}
\num\newcommand{\RackC3}{Rack C3} % ~\cite{DBLP:conf/splc/Comploi-TaupeFS22}}
\num\newcommand{\RackC21}{Rack C7} % ~\cite{DBLP:conf/splc/Comploi-TaupeFS22}}

\num\newcommand{\HCP1}{HCP 1} % ~\cite{DBLP:conf/confws/FriedrichRFHSS11}}
\num\newcommand{\HRP1}{HRP 1} % ~\cite{DBLP:conf/confws/FriedrichRFHSS11}}
\num\newcommand{\HCP3}{HCP 3} % ~\cite{DBLP:conf/confws/FriedrichRFHSS11}}
\num\newcommand{\HRP3}{HRP 3} % ~\cite{DBLP:conf/confws/FriedrichRFHSS11}}
\num\newcommand{\Monkey1}{Monkey 1} % ~\cite{DBLP:conf/sat/Reger0V16}}
\num\newcommand{\Monkey4}{Monkey 4} % ~\cite{DBLP:conf/sat/Reger0V16}}
    \newcommand{\PUP}{PUP} % ~\cite{DBLP:journals/aiedam/FalknerHSS08}}
\num\newcommand{\CADE071}{CADE07-1} % ~\cite{DBLP:conf/vmcai/SuterSK11}}
\num\newcommand{\CADE0722}{CADE07-2b} % ~\cite{DBLP:conf/vmcai/SuterSK11}}
\num\newcommand{\CADE0760}{CADE07-6} % ~\cite{DBLP:conf/vmcai/SuterSK11}}
\num\newcommand{\CADE0762}{CADE07-6b} % ~\cite{DBLP:conf/vmcai/SuterSK11}}

\begin{table}[ht]
    \small
    \center
    \caption{Wall clock time in seconds to solve configuration problems, and number (\#) of used domain elements in models. Yes: lifted sentence, No: original sentence, T: timeout after 200 sec.}
    \label{tab-results}
    \begin{threeparttable}
    \begin{tabular}{l|cc|cc}
        ~ & \multicolumn{2}{c|}{Time} & \multicolumn{2}{c}{\#} \\
        \multicolumn{1}{r|}{Lifted?} & Lifted &  Orig. & Lifted & Orig.\\
        \hline
        %\textbf{Problem} \\
        \hline
        \HCP1 & \textbf{0.2} &  \textbf{0.2}   & 6 & 10 \\
        \HCP3 & \textbf{0.22} &  1.5   & 6 &  27 \\
        \HRP1 &  0.27 & \textbf{0.23} &  7 & 10 \\
        \HRP3 &  \textbf{0.30} & 6.41 & 8  & 25 \\
        \Monkey1 &  \textbf{0.28} & 184.66    & 6  & 20 \\
        \Monkey4 & \textbf{0.27} &  T     & 6 & 52 \\
        \Rack2011 & \textbf{0.24}   & 8.82 & 5 & 23 \\
        \RackA5 &  \textbf{0.46} & 0.86 &  5 & 19 \\
        \RackA10 & \textbf{0.41} &  37.52    & 5 & 38 \\
%        \RackA50 & \textbf{0.82}  &    T  & 5 & 125 \\
        \RackA1000 & \textbf{0.54}  &   T  & 4 & 2250 \\
        \RackABCD4 & \textbf{1.55}   & 4.37  & 13& 20 \\
        \RackABCD8 &  \textbf{2.69}  &  T     & 13 & 40 \\
        %\RackABCD12 &  \textbf{2.77}  &  T     & 13  & 60 \\
%        \RackABCD16 &   \textbf{2.78} &  T    & 14 & 80 \\
        \RackABCD20 &   \textbf{3.31}   &  T      & 13 & 100 \\
        %\RackABCD20b &  T   &  T   & \textbf{5.76} &  T   & ~ &  T   &  T   &  T   &  T   & ~ \\
        %\RackC21 & 1.02 & 1.15 & 0.73 & \textbf{0.7}  &  T   &  T   &  T   &  T  & 5 & 63 \\
        \hline
        \RackA1 & \textbf{0.61} &   0.75   & 5 & 7 \\
        \RackA2 & \textbf{0.64} & 1.54 &   5 & 9 \\
        \RackA3 & \textbf{0.62} &   3.69 &   5 & 11 \\
        \RackA4 & \textbf{0.65} &   36.9 &   5 & 13 \\
        \RackA5 & \textbf{0.65} &   T &   5 & 19 \\
        % \RackABCD4 & T &  & - &   & 13 & T &  & - & & 20 \\
        % \RackC1 & 0.77 &  & - &   & 5 & 1.5 &  & - & & 9 \\
        % \RackC2 & 0.81 &  & - &   & 5 & 27.17 &  & - & & 13 \\
        % \RackC3 & 0.83 &  & - &   & 5 & T &  & - & &  18 \\
    \bottomrule
    \end{tabular}
    % \begin{tablenotes}[flushleft]
    %     \footnotesize
    %     \item[(1)] https://google.github.io/or-tools/, used  via CPMpy~\cite{guns2019increasing}.
    %     %\item[(2)] Some runs timed out, due to nondeterminism of the unsat-core algorithm in Z3.
    % \end{tablenotes}
    \end{threeparttable}
\end{table}

The top half of Table~\ref{tab-results} shows results with Z3 (with automatic translation of the sentence), the bottom half with the OR-tools solver\footnote{https://google.github.io/or-tools/, used  via CPMpy~\cite{guns2019increasing}.} (with manual translation).
A more detailed table is provided in Appendix B. %anon

Solutions to problems with higher suffixes have more components than similar problems with lower ones.
The table shows near constant time performance on the Rack ABCD problems.
The lifted methods solve each of the 20 occurrences of the ABCD Racks problem in \cite{DBLP:conf/splc/Comploi-TaupeFS22} in less than 5 seconds (instead of 6 minutes on average in that paper, using an ASP solver).
These results show that our method has significant merits for solving problems with symmetries and a preponderance of unary symbols.

    \section{Boolean Algebra of sets with Presburger Arithmetic}\label{sec:BAPA}

    Lifted domain elements represent disjoint sets of concrete domain elements.
    A model search in the lifted domain can be seen as a model search involving sets.
    Hence, our work is highly related to Boolean Algebra of sets with Presburger Arithmetic (BAPA), a logic that can express constraints on the cardinality of sets, of their unions and of their intersections~\cite{DBLP:conf/cade/KuncakR07, DBLP:conf/vmcai/SuterSK11, bansal2016new}.
    Some problems from verification of properties of software operating on complex data structures contain fragments that belong to BAPA.

    A sample BAPA statement is $|A| > 1 \land A \subseteq B \land |B \cap C | \leq 2$, where $A, B, C$ are sets, and $|A|$ is the cardinality of $A$.
    The equivalent expression in \fota~is
    $(\#\{d: A(d)\} > 1) \land (\forall d : A(d) \Rightarrow B(d)) \land (\#\{d: B(d) \land C(d)\} \leq 2)$, where $A, B, C$ are now unary predicates over a (unique) type whose interpretation is to be found.
    This expression can be lifted and solved using our approach (see the ``theories/BAPA'' folder in our repository\footnoteref{repository}). %\maurice{update?}
    In general, any BAPA sentence can be converted to a concrete \fota~sentence that only uses unary predicates, and that can be lifted without regularity constraints.
    Hence, our approach offers a simple way to solve BAPA problems using any solver capable of reasoning over the rationals.
    The performance advantage should be evaluated in future work.

    On the other hand, the conversion of a concrete \fota~sentence to BAPA logic is challenging because \fota~is more expressive:
    it allows n-ary relations, functions, sum aggregates, and product of cardinalities.
    Thus, a BAPA solver could not be used to solve, e.g., generative configuration problems.
    Still, extensions of BAPA solvers to handle finite n-ary relations have been implemented in CVC4~\cite{DBLP:conf/cade/MengRTB17}.
    %However, there is no known method to expand BAPA lifted models to concrete models, as we propose.

    A simple approach to represent structures in BAPA is to use disjoint subsets of the concrete domain, called \emph{Venn regions}, so that the cardinalities of any set of interest is the sum of the cardinalities of its Venn regions.
    Unfortunately, the number of Venn regions grows exponentially with the number of sets of interest.
    Hence, various methods have been developed to reduce this growth, e.g., by creating new Venn regions lazily when required~\cite{bansal2016new}.
    Venn regions are similar to our lifted domain elements, and the iterative method that creates new Venn region is similar, to some extent, to our iterative method that creates new lifted elements when required.

    Our method cannot prove the unsatisfiability of a sentence.
    By contrast, efficient methods have been proposed to prove the unsatisfiability of BAPA sentences~\cite{DBLP:conf/vmcai/SuterSK11}.
    %Sadly, their implementation in Z3 is not supported anymore.\footnote{https://github.com/Z3Prover/z3/issues/3854}

    % \subsection{Domain lifted probabilistic inferences}

    % In the field of probabilistic inferences, an algorithm is called domain lifted if it runs in time polynomial in the size of the domain~\cite{DBLP:conf/ijcai/BrazAR05}.
    % The primary concern in that field is to count the number of models of a sentence, and efficient domain lifted algorithms have been developed for that purpose.
    % By contrast, our purpose is to find a model of a sentence.

    \section{Discussion}
    \label{section:conclusion}

    Unlike the traditional approach of adding symmetry breaking conditions to a formula to accelerate satisfiability checking, we \emph{automatically translate} the formula to a form with fewer symmetries.
    Our results demonstrate the benefits of this approach for problems with symmetries and a preponderance of unary symbols, and justify further research in automatic translation for symmetry reduction.

    % To the best of our understanding, our method is very different from current symmetry breaking methods in SAT solving, from methods for model counting and from methods for lifted inference in probabilistic logic (although we borrowed the term "lifted" from that domain).

    Much work remains to be done in evaluating the method and determining problem areas where it is useful. Moreover, the relationship with BAPA logic is worth further exploration because, unlike our method, BAPA can identify unsatisfiable sentences and our method does not terminate for such problems.
    Efficient implementation of the $lcm$ function is another area of research.
    Worth exploring is also the relevance for related areas such as model counting.
    %Also, our solver should be extended to cope better with symmetries in the lifted domain (with many unneeded lifted domain elements, the performance goes down). Ordering the lifted domain elements $l_1,l_2,\ldots $ of a type and adding the constraint $\mul(l_1) \geq \mul(l_2) \geq \ldots$ is a first step. %If so, each iteration could introduce several extra lifted domain elements and less iterations would be needed to find a  solution that minimizes the use of some resource (e.g., the number of holes occupied by a pigeon).

    Ideally, any compression of a model of a sentence (Section~\ref{sec:lifted-model}) should be a model of the lifted sentence. It is unlikely that this is achievable, as it requires a translation that avoids all regularity constraints apart from those for functions. This requires being able to predict the fraction of concrete variable assignments (in the expansion of a lifted assignment) that make a concrete formula true, given that the translated formula is true with the lifted assignment.
    %for every translated formula with free variables $\bar x$ that is true in a lifted structure under variable assignment $[\bar x: \bar l]$.
    In particular, filter formulas (in an aggregate) having both free and quantified variables are problematic.

    Still, many refinements of the translation in Table~\ref{table:trans} are feasible. We have already worked out some refinements, but they cannot be presented within the current space constraints.
    They will be presented in the PhD thesis of one of the authors~\cite{carbonnellePhD}.

\bibliography{references}

\newpage

% This document presents the Supplementary Material for the paper ``Using Symmetries to Lift Satisfiability Checking''.  
% An incorrect Appendix was prematurely included in the submitted version of the paper. 
% This document contains the correct Appendices.

% Appendix A provides a proof that was too long to include in the paper.  
% Appendix B provides details about the experimental results for Generative Configuration Problems.
% The source code is provided in a separate zip file.

\section{Appendix A}\label{sec:proof}

In this section, we prove that the expansion of a model of the translated sentence is a model of the original sentence with the given size of the concrete types. 

In what follows, we use $L$ to denote a lifted model of the translation of the sentence to be satisfied, and $I$ to denote its expansion, whose automorphism $\pi$ has the domain of $L$ as backbone.
$L$ is a regular structure because it satisfies the regularity condition for functions (part (ii) of the translated sentence); 
$I$ has the given size of the concrete types because it is the expansion of $L$ which satisfies the sentences stating those sizes in the translated sentence (part (iii) of the translated sentence). 
We use $l$ for lifted domain elements from $L$ and $d$ for domain elements from $I$.

We begin with a definition.

\begin{definition}[Expansion of a variable assignment]\label{def:expassign}
    The expansion of a lifted variable assignment $[\bar x:\bar l]$  (notation $\Exp([\bar x:\bar l])$) over the domain of  $L$ is the set of concrete variable assignments $[\bar x: \bar d]$ over the domain of $I$ such that $\bard \in (\bigtimes_i \Exp(l_i))$.
\end{definition}

We say that $[\bar x: \bar d]$ is an \emph{instance} of $[\bar x:\bar l]$.
Note that there are $\Mul(\bar x)$ instances of $[\bar x: \bar l]$, but only $\LCM(\bar x)$ tuples in the expansion of $\bar l$. 
Thanks to the $0 < \Mul(\barx)$ filters in the transformation of quantified formulas and aggregates, we consider only the lifted variable assignments that have instances, unless otherwise stated.

%We only consider terms and formulas occurring in the sentence to be satisfied, or its translation.
By abuse of language, we use $L$ and $I$ to denote a structure extended with a variable assignment: when we say that an expression $\tr(e)$ is evaluated in $L$, and it is in the scope of quantifiers or sum aggregates that collectively declare variables $\bar x$, we mean that it is evaluated in $L[\bar x: \bar l]$;
when we say that $e$ is evaluated in $I$, we mean that it is evaluated in $I[\bar x:\bar d]$, with $[\bar x: \bar d]$ an instance of $[\bar x:\bar l]$. 
%We use $l_t$ to denote the value of a term $\chi(t)$ in $L$ (the result of evaluating $\chi(t)$ in $L[\bar x: \bar l]$) and $d_t$ to denote the value of the term $t$ in $I$. 

The theorem is a direct consequence of the following two related properties:

\begin{proposition}[Preservation of terms]\label{prop:hypothesis-term}
    Let $\tr(t)$ be the transformation of a term $t$ occurring in the sentence to be satisfied.
    There exists an $i$ such that $t^I = \pi^i(\chi(t)^L)$.
\end{proposition}

In words, the value of a term $t$ in structure $I$ is part of the
expansion of the value of $\tr(t)$ in $L$, for any variable assignment in $I$ that is an instance of the variable assignment in $L$.
%, given that, in the presence of free variables, the variable assignment in $I$ is an instance of the variable assignment in $L$.
Note that a number is mapped to itself by $\pi$:  Proposition~\ref{prop:hypothesis-term} implies that the value of a numeric term in $I$ is the value of the transformed term in $L$.

\begin{proposition}[Preservation of formulae]\label{prop:hypothesis-formula}
    Let $\tr(\F)$ be the transformation of a formula $\F$ occurring in the sentence to be satisfied.
    $\tr(\F)$ is true in $L$ if and only if $\F$ is true in $I$. 
\end{proposition}

In words, $\F$ in $I$ has the truth value of $\tr(\F)$ in $L$, for any variable assignment in $I$ that is an instance of the variable assignment in $L$. %, given that, in the presence of free variables, the variable assignment in $I$ is an instance of the variable assignment in $L$.

Note that Proposition~\ref{prop:hypothesis-formula} implies our goal: the expansion of a
model of the translated sentence is a model of the original sentence.  In fact,
an implication in Proposition~\ref{prop:hypothesis-formula} would be enough: this can be
exploited to weaken the constraints in the translation of top-level sentences
(but there is not enough space to do so in this paper). 
%We state Proposition~\ref{prop:hypothesis-formula} with an equivalence because an equivalence is needed to prove the property for sum aggregates.

The proof of these two properties is by structural induction over the syntactic tree of the sentence to be satisfied.
The base cases are the leaves of the syntactic tree; for the inductive step, we prove that the property holds in a node of the syntactic tree when it holds in all the branches of the node.
We consider every language construct and its translation per Table~1. %\ref{table:trans}.
In what follows, we consider only terms and formulas occurring in the sentence to be satisfied.

\paragraph{Simple cases}
Proving the base cases is trivial (Rules~14\ignore{eq:Xx}, 15\ignore{eq:Xc}, 20\ignore{eq:Xprop}): the value of $x$ in $I$ is in the expansion of its value in $L$, and this expansion forms a cycle of $\pi$; the value of a constant $c()$ is the same in both structures because their multiplicity is 1; the truth value of $p()$ is also the same by construction of $p^I$.

Proving the correctness of the inductive step for arithmetic operations and comparisons, as well as for the boolean connectives (Rules~17\ignore{eq:Xop}, 22\ignore{eq:Xeqa}, 23\ignore{eq:Xand}, 24\ignore{eq:XNot}) is also trivial because the value of their numeric sub-terms and boolean sub-formulas are the same in both structures, by the induction hypothesis and because the multiplicity of numbers is 1.

%The inductive step for predicate and function application (Rule~\ref{eq:Xf}) follows from the fact that the function is regular by the regularity condition in the translated sentence, and from the definition of the expansion of a function.

\paragraph{Function application}

We now consider the case where the term $t$ is of the form $f(\bar t)$ (Rule~16\ignore{eq:Xf}).

By the induction hypothesis, for each argument $t_j$, there is an $i_j$ so that $t_j^I = \pi^{i_j}(\chi(t_j)^{L})$.
Thus, $\chi(t_j)^{L}$ is the (unique) backbone element that generates $t_j^I$.
Because structure $L$ is regular, the interpretation of function $f$ is regular: all the elements in its concrete domain can be obtained by repeated application of $\pi$ to a backbone tuple in its domain.  
The backbone tuple for $\bar t^I$ is $(\chi(t_1)^{L},\dots,\chi(t_n)^{L})=\chi(\bart)^{L}$.
%Here, the concrete tuple of values $\bart$ in $I$ is a permutation of the lifted tuple of backbone values $(l_{t_1}),\dots,l_{t_n})=l_{\bar t}$ in $L$: $\bar t^I = \pi^i(\bart^L)$ for some $i$.
Then, \\
$t^I = f^I(\bart^I)$ by definition of $(f(\bart))^I$ \\
$= f^I(\pi^i(\chi(\bart)^{L}))$ for some i, by above argument\\
$ = \pi^i(f^I(\chi(\bart)^{L}))$ by automorphism of $\pi$\\
$= \pi^i(f^L(\chi(\bart)^{L}))$ by property of backbone \\
$= \pi^i(\chi(t)^L)$
%$t^I = f^I(\bart^I) = f^I(\pi^i(\bar t)) = \pi^i(f^I(\bar t)$.

\paragraph{Predicate application} To prove the correctness of the inductive step for predicate applications (Rule~21\ignore{eq:Xp}), we start with a lemma about variable assignments.

\begin{lemma}\label{prop:tuple-expansion}
    The value of $\bar t$ in $I$ for the variable assignment  $[\bar x: \pi^i(\bar l)]$ is $\pi^i(\chi(\bar t)^{L[\barx: \barl]})$, for any $i$.
\end{lemma}
This can be easily proved for each $t_i$ by structural induction over its syntactic structure.

\begin{lemma} \label{prop:var-exp} 
    If $\Mul(\bar x) = \LCM(\bar x)$ holds in $L[\bar x:\bar l]$,
    %\todo{I would use $\barl$ instead of $\barx$  maurice: No because the translation table has to use $\bar l$.  Pierre:NO.  This proposition stands independently of the translation table.} 
    then, for each instance $[\bar x:\bar d]$ of $[\bar x:\bar l]$, there exists an $i$ such that $\bar d= \pi^i(\bar l)$.
\end{lemma}

The lemma holds because the expansion of $\bar l$ has $\LCM(\bar l)= \MUL(\bar l)$ tuples, so $\bar l$ is regular: each tuple in the expansion of $[\bar x: \bar l]$ assigns a value in the expansion of $\bar l$.

The next lemma states some properties of atoms.

\begin{lemma}\label{prop:atomexpansion}
    Let $p(\bar t)$ be an atom with free variables in $\bar x$. 
    The following properties hold:
    \begin{description}
            
        \item [i] If $p(\tr(\bar t))$ is true in $L$, then $p(\bar t)$ is true in $I$ for at least $\LCM(\bar x)^L$ instances of the lifted variable assignment and  false for at most $(\Mul(\bar x) - \LCM(\bar x))^L$ instances.
    
        \item [ii] If $p(\tr(\bar t))$ is true in $L$ and either $\MUL(\chi(\bar t))=\LCM(\chi(\bar t))$ or $\MUL(\bar x)=\LCM(\bar x)$ is true in $L$, then $p(\bar t)$ is true in $I$ for all $\Mul(\bar x)^L$ instances of the lifted variable assignment.
    
        \item [iii] If $p(\tr(\bar t))$ is false in $L$, then $p(\bar t)$ is false in $I$ for all $\Mul(\bar x)^L$ instances of the lifted variable assignment.

   \end{description}
\end{lemma}

% In case there are no free variables, Lemma~\ref{prop:atomexpansion} (which states the preservation of their values) holds because all terms of $\bar t$ have multiplicity 1 and their value is preserved by the induction hypothesis. 
Property (i) follows from Lemma~\ref{prop:tuple-expansion}, from the fact that $p(\bart) = p(\pi(\bart))$ in $I$ by automorphism of $\pi$, and from the fact that there are $\LCM(\bar x)^L$ distinct permutations of $[\bar x: \bar l]$.
Furthermore, there are $\MUL(\bar x)^L$ distinct instances of the lifted variable assignments.
$\MUL(\tr(\bar t))=\LCM(\tr(\bar t))$ in $L$ implies that every possible value of $\bar t$ in $I$ is in the expansion of the value $\chi(\bart)^L$ and $\MUL(\bar x)=\LCM(\bar x)$ in $L$ implies (Lemma~\ref{prop:var-exp})  that every variable assignment $[\bar x:\bar d]$ is a permutation of $[\bar x:\bar l]$. Both conditions imply (ii). 
Property (iii) holds because of the definition of the expansion of a predicate (Equation~11\ignore{def:expansion-predicate}). %the expansion of different domain elements cannot overlap.

Note that, when the conditions in (ii) hold, properties (ii) and (iii) together imply that the preservation of value holds for the atom.  The condition in (ii) is the regularity condition that is added to the translation of an atom by Rule~21\ignore{eq:Xp}. 
Since $L$ satisfies this condition, the value for predicate application is preserved (Rule~21\ignore{eq:Xp})
%For atoms that match Rule~\ref{eq:Xeqa}, the constraint is not needed because both terms are numbers and have multiplicity 1. 

% A straightforward consequence of the Proposition~\ref{prop:atomexpansion} is:
% \begin{proposition}\label{prop:formula}
%     Let $\F$ be a formula.
%     If the term induction hypothesis holds for every aggregate term in $\F$ and the formula induction hypothesis for every quantified subformula of $\F$ and if, for every atom $p(\bar t)$ in $\F$ with free variables $\bar x$ the constraint $\MUL(\tr(\bar t))=\LCM(\tr(\bar t)) \lor \MUL(\bar x)=\LCM(\bar x)$ is added to the translation of the atom, then the formula induction hypothesis holds for the translation of the formula. 
% \end{proposition}

\paragraph{Quantification}  
The translation of $\exists \bar x: p(\bar t[\bar x], s[\bar y]) \land \F$ (Rule~26\ignore{eq:Xexists2}) adds the atom $0<\Mul(\bar x)$ to the conjunction. If $\Mul(\bar x) =0$ in $L$, the sub-formula is false in $L$; it is also false in $I$ because, in $I$, there are no tuples instantiating $\bar x$. 

So, we are left with the case that $0< \MUL(\bar x)$.  
Let's denote by $\psi$ the formula $p(\chi({\bar t}),\chi(s)) \land \chi(\phi)$.

If $\psi$ is false in $L$ for some lifted variable assignment, then it is false in $I$ for all variable assignments that are instances of the lifted one. Indeed, if the atom $p(\tr(\bar t),\tr(s))$ is false in $L$, it follows from Lemma~\ref{prop:atomexpansion}-(iii) that the atom is false in $I$; if, on the other hand, $\tr(\F)$ is false in $L$, it follows from the induction hypothesis that $\F$ is false in $I$.

If $\psi$ is true in $L$ for some lifted variable assignment $[\bar x: \bar l_x :: \bar y:\bar l_y]$ ($\bar y$ are the variables of $s$; they do not occur in $\bar x$).
Consider an instance $[\bar y:\bar d_y]$. By the induction hypothesis, $s$ evaluates for this instance in $I$ to $\pi^i(\chi(s)^L)$ for some $i$.
By Lemma~\ref{prop:tuple-expansion}, $\bar t$ evaluates for the instance $[\bar x: \pi^i(\bar l_x)]$ to $\pi^i(\chi(\bar t)^L)$.
The tuple $\pi^i(\chi(\bar t)^L,\chi(s)^L)$ is in the expansion of the interpretation of $p$ by automorphism $\pi$ and by hypothesis on $\psi$, and so is true in $I$. Because $\F$ is also true for this variable assignment by the induction hypothesis, $\exists \bar x:  p(t[\bar x], s[\bar y]) \land \F$ is then also true in $I$.

The proof for the general case (Rule~28\ignore{eq:Xexists}) follows the same reasoning (without the $p$ atom).

Rule~25\ignore{eq:Xforall2}, the special rule for a universally quantified formula, is proved by transforming the universal quantification to an existential quantification using the law $(\forall \barx: \phi) \equiv \lnot(\exists \barx:\lnot\phi)$, by applying Rule~26\ignore{eq:Xexists2}, and transforming the resulting existential quantification to a universal quantification. 
%f the translation of the formula is true for some lifted variable assignment, then its negation is false, and if it is false, then the negation is true. But the negation is $\exists x: 0<\Mul(\bar x) \land p(\tr(\bar t),\tr(s)) \land \lnot \F$ which is the translation of the negation of the formula. This negated formula is an instance of the rule in the previous item.
A similar reasoning as above holds for the general rule of universal quantification (Rule~27\ignore{eq:Xforall}).

\begin{table*}[!t]
    \small
    \center
    \caption{Wall clock time in seconds to solve configuration problems, and number (\#) of used domain elements in models. The top half shows results with Z3, the bottom half with OR-tools. T: timeout after 200 sec.}
    \label{tab-results2}
    \begin{threeparttable}
    \begin{tabular}{l|cccc|cccc|cc}
        ~ & \multicolumn{8}{c|}{Performance} & \multicolumn{2}{c}{\#} \\
        Lifted? & \multicolumn{4}{c|}{Yes (our approach)} &  \multicolumn{4}{c|}{No} & Yes & No\\
        Typed? & no &no &yes &yes & no &no&yes&yes&  & \\
        Domain change & $+1$&$\times 1.5$& $+1$&$\times 1.5$ &$+1$&$\times 1.5$& $+1$&$\times 1.5$&   &  \\
        Method name & $lm_1$ & $lm_2$ & $lt_1$ & $lt_2$ &$m_1$ & $m_2$ & $t_1$ & $t_2$ &  &  \\
        \hline			\hline
        \textbf{Problem} &&&&&&&&&&\\
        \hline

        \hline
        % 16-8-2023
        \HCP1 & 0.55 & 0.47 & \textbf{0.2} & \textbf{0.2} & 1.27 & 1.00 & 0.22 & \textbf{0.2} & 6 & 10 \\ 
        \HCP3 & 0.55 & 0.46 & 0.22 & \textbf{0.2} &  T   & 3.60 & 1.50 &  T   & 6 & 27 \\ 
        \HRP1 & 1.13 & 1.17 & 0.27 & \textbf{0.26} & 1.93 & 1.60 & 0.34 & \textbf{0.23} & 7 & 10 \\ 
        \HRP3 & 1.19 & 1.34 & 0.30 & \textbf{0.29} &  T   & 6.41 &  T   &  T   & 8 & 25 \\ 
        \Monkey1 & 0.55 & 0.45 & 0.28 & \textbf{0.25} & 184.66 &  T   &  T   &  T   & 6 & 20 \\ 
        \Monkey4 & 0.54 & 0.48 & 0.27 & \textbf{0.26} &  T   &  T   &  T   &  T   & 6 & 52 \\ 
        \Rack2011 & 0.82 & 0.78 & \textbf{0.24} & \textbf{0.24} & 23.88 & 19.43 &  T   & 8.82 & 5 & 23 \\ 
        \RackA5 & 0.93 & 1.09 & 0.46 & \textbf{0.45} & 10.99 & 6.74 & 0.86 & 2.48 & 5 & 19 \\ 
        \RackA10 & 0.97 & 1.06 & \textbf{0.41} & 0.42 & 126.14 & 37.52 &  T   &  T   & 5 & 38 \\ 
        \RackA1000 & 0.59 & 0.61 & \textbf{0.54} & \textbf{0.54} &  T   &  T   &  T   &  T   & 4 & 2250 \\ 
        \RackABCD4 & 20.71 & 5.69 & \textbf{1.55} & 3.34 & 118.49 & 106.66 & 24.62 & 4.37 & 13 & 20 \\ 
        \RackABCD8 & 165.58 & 9.92 & \textbf{2.69} & 5.10 &  T   &  T   &  T   &  T   & 13 & 40 \\ 
        \RackABCD20 &  T   & 16.41 & \textbf{3.31} & 7.87 &  T   &  T   &  T   &  T   & 13 & 100 \\ 
        \hline
        \RackA1 & \textbf{0.61} &  &  &  & 0.75 &  &  &  & 5 & 7 \\
        \RackA2 & \textbf{0.64} &  &  & & 1.54 &  &  &  & 5 & 9 \\
        \RackA3 & \textbf{0.62} &  &  &  & 3.69 &  &  & & 5 & 11 \\
        \RackA4 & \textbf{0.65} &  &  & & 36.9 &  &  &  & 5 & 13 \\
        \RackA5 & \textbf{0.65} &  &  &  & T &  &  & & 5 & 19 \\
        % \RackABCD4 & T &  & - &   & 13 & T &  & - & & 20 \\
        % \RackC1 & 0.77 &  & - &   & 5 & 1.5 &  & - & & 9 \\
        % \RackC2 & 0.81 &  & - &   & 5 & 27.17 &  & - & & 13 \\
        % \RackC3 & 0.83 &  & - &   & 5 & T &  & - & &  18 \\
    \bottomrule
    \end{tabular}
    % \begin{tablenotes}[flushleft]
    %     \footnotesize
    %     %\item[(1)] \url{https://google.github.io/or-tools/}  via CPMpy~\cite{guns2019increasing} .
    %     %\item[(2)] Some runs timed out, due to nondeterminism of the unsat-core algorithm in Z3.
    % \end{tablenotes}
    \end{threeparttable}
\end{table*}

\paragraph{Sum aggregate}

We begin with the translation of $\SUMW{\bar x\in \bar T}{\F}{t}$ (Rule~19\ignore{eq:Xsum}). 
By the induction hypothesis, $\F$ is true in $I$ for every instance of a lifted variable assignment for which the formula is true in $L$ and so, the formula is true in $I$ for all $\Mul(\bar x)$ instances. The formula $\F$ is false in $L$ for every instance of a lifted variable assignment for which the formula is false in $L$. As $\tr(t)$ in $L$ and $t$ in $I$ have the same value by the induction hypothesis for numeric terms, it follows that the term $\Mul(\bar x) \times \tr(t)$ ensures that the translated formula evaluates in $L$ to the same value as the original formula in $I$.

For the translation of $\SUMW{\bar x\in \bar T}{p(\bar t[\barx],s[\bar y]) \land \F}{t}$ (Rule~18\ignore{eq:Xsumf2}), given the proof for Rule~19\ignore{eq:Xsum}, it suffices to show that $\Mul(\bar x) \times \LCM(\tr(\bar t),\tr(s))\div  (\LCM(\tr(\bar t)) \times \mul(\tr(s)))$ is the number of instances of $[\bar x:\bar l_x]$ for which the filter formula is true in $I$ when the translated filter formula is true in $L$ for $[\bar x: \bar l_x:: \bar y: \bar l_y]$ and $\bar y$ is assigned a value $\bar d_y$. 

% The number of instances of $[\bar x:\bar l_x]$ is $Mul(\barx)$.  
% This number must be multiplied by the fraction of instances of $(\bar t, s)$ that make $p$ true in $I$.
% The number of instances of $(\bar t, s)$ is $\LCM(\chi(\bart)) \times \mul(\chi(s)))$ because of the independence between $\bar t$ and $s$;
% The number of instances that make $p$ true in $I$ is $\LCM(\chi(\bar t),\chi(s))$.  % not correct by Lemma 3
% Thus, the fraction is $\LCM(\chi(\bar t),\chi(s)) \div (\LCM(\chi(\bart)) \times \mul(\chi(s))))$.
    
Let's call $(\bar l_t,l_s)$ the value of $(\chi(\bart), \chi(s))$ in $L$ for the variable assignment $[\bar x:\bar l_x::\bar y:\bar l_y]$.
Consider a tuple $(\bar l_t,l_s)$ that is true in the lifted interpretation of $p$.
The expansion of this (true) tuple has $\LCM(\chi(\bar t),\chi(s))$ (true) tuples by the definition of the expansion of $p$.
In $I$, the term $s$ for a variable assignment $[\bar y:\bar d_y]$ evaluates to $d_s=\pi^i(l_s)$ for some $i$ (by the induction hypothesis) and there are $\mul(l_s)$ different values in the expansion of $l_s$.
So, there are $\LCM(\chi(\bar t),\chi(s))\div \mul(\chi(s))$ true tuples $(\bar d_t,d_s)$ in the expansion of  $(\bar l_t,l_s)$ when $\bar y$ is assigned a value $\bar d_y$. %such that $d_s=\pi^i(l_s)$.

There are $\Mul(\bar x)$ different variable assignments to $\bar x$, and there are $\LCM(\chi(\bar t))$ tuples in the expansion of $\bar l_t$, so each value $\bar d_t$ in the expansion of $\bar l_t$ is produced by $\Mul(\bar x) \div \LCM(\chi(\bar t))$ different  variable assignments.

Hence, the number of variable assignments $[\bar x:\bar d_x]$  such that a true tuple $(\bar d_t,d_s)$ is obtained is  $\Mul(\bar x) \LCM(\tr(\bar t),\tr(s))\div  \LCM(\tr(\bar t)) \mul(\tr(s))$. This is a whole number as $\Mul(\bar x)$ is a multiple of $\LCM(\chi(\bar t))$ and $\LCM(\chi(\bar t),\chi(s))$  is a multiple of $\mul(\chi(s))$.

This completes the proof by structural induction.

\section{Appendix B}

This section provides more details on our results for Generative Configuration Problems.
Table~2\ignore{tab-results} in the paper is an excerpt of the more detailed Table~\ref{tab-results2}.  
Best results are shown in bold.

Our baseline is a portfolio solver using 4 traditional iterative methods to solve Generative Configuration Problems.  To generate the four solvers, we automatically transform a concrete sentence describing a configuration problem in $\mathrm{FO}(\cdot)$ into four equivalent concrete formulations of the sentence. % for use with four traditional iterative methods of satisfiability checking over an unknown domain.
The transformation consists of adding new symbols to the vocabulary (e.g., \ConfigObject), and of adding additional constraints on those symbols to the formula.
The four formulations are denoted by a two-letter word, e.g., ``$m_1$'':
\begin{itemize}
    \item the first character distinguishes between typed (``t'') and monotype (``m'', using a single type) formulations. For a typed formulation, the unsatisfiable core of the sentence is used to determine which type(s) needs a larger domain.
    %whose type(s) the assumed number of domain elements must be increased;
    \item the second character says that the assumed number of elements in a type is incremented by 1 (``1'') at each iteration, or multiplied by 1.5 (``2'').
\end{itemize}

% \begin{itemize}
%     \item (``$m_1$'') The vocabulary has only one type (``ConfigObject''), and the exact number of domain elements is initially assumed to be 1; this number is incremented until a model is found. %\bbcomment{You cannot increment assumptions. The number of objects is incremented by one? Same for items below.}
%     \item (``$m_2$'') The vocabulary has only one type, and the upper bound for the number of domain elements is initially assumed to be 1; this number is multiplied by 1.5 until a model is found; the interpretation of a predicate determines which domain elements represent actual component(s) of the configuration.
%     \item (``$t_1$'') The vocabulary has a type for each type of components in the configuration; the upper bound of the size of each type is initially assumed to be 1; this number is incremented for the types used in the unsatisfiable core of the formula,  until a model is found; the interpretation of a predicate determines which domain elements represent actual component(s) of the configuration.
%     \item (``$t_2$'') Similar to ``$t_1$'', but the assumption is multiplied by 1.5.
% \end{itemize}

To evaluate our method, we automatically translate the concrete sentence into a lifted version per our method, and automatically apply the 4 transformations above, for use in lifted iterative model searches.
The names of the lifted model search are prefixed by ``$l$'' (e.g., $lm_1$).
The number of lifted elements is found iteratively, and their multiplicities are searched by the solver in each iteration, until the lifted sentence is satisfiable.
Finally, we verify that an expansion of the models of the lifted sentence satisfies the concrete sentence. % (this step is not included in time measurements).

Table~1 is generated by comparing the solving time of lifted method $lt_1$ to the best solving time of $m_1, m_2, t_1$, and $ t_2$.
The performance of $lm_1, lm_2, lt_2$ are given only for information.
The number of elements (two leftmost columns) is obtained by manual counting in the obtained models.

Note that the iterative methods (except $m_1$) may result in the domain having more elements than needed for a solution.
For example, a model found with $m_2$ may have 9 domain elements, even though the solution requires only 7; two are ``unused''.
Thus, the concrete vocabularies for the iterative methods (except $m_1$) have a predicate to mark the used domain elements, and the quantification/aggregates in the theories are adapted to consider the used domain elements only.
Similarly, in the lifted domain, unused elements have multiplicity zero.

Because the unsatisfiable core algorithm used by Z3 is not deterministic, the types that are enlarged at each iteration may vary from one run to another, leading to potentially different performance results.  The durations reported in the table are representative.

% \section{Appendix}

% See AAAI.tex file

\end{document}

\maketitle

\section{Appendix: proofs} %%%%%%%%%%%%%%%%%%%%%%%%%%%%%%%%%%%%%%%%%%%%%
\label{Appendix:proofs}

\includeExternalAppendix{main}

\begin{proposition} \label{prop:tuple}
    Let $L$ be a lifted structure that satisfies the regularity constraints in the lifted theory $\overline{\chi(\Th)}$ (third column of Table~\ref{table:trans}) and $I$ its expansion.
    Let $\bar t$ (resp. $\bar t_L$) be a tuple of terms occurring in the concrete theory $\Th$ (resp. its translation).
    Let $\bar x$ be the free variables of $\bar t$ and $[\bar x:\bar l]$ a lifted variable assignment.
    Let $[\bar x: \bar d]$ be a variable assignment in the expansion of the lifted one.
    If $\Reg(\bar t_L^{L[\bar x:\bar l]})$ then there exists an $i$ such that $\bar t^{I[\bar x:\bar d]}=\pi^i(\bar t_L^{L[\bar x:\bar l]})$.
\end{proposition}

\begin{proof}
    %The regularity  of $\bar l$ implies there exists an $i$ such that $\bar d = \pi^i(\bar l)$. So, it suffices to prove that, for all $j$, $t_j^{I[\bar x:\bar d]}= \pi^i(\bar {t_j}_L^{L[\bar x:\bar l]})$.
    We prove the lemma that, for any concrete tuple $\bard$ in the expansion of lifted tuple $\bar l$, $\bar t^{I[\bar x:\bar d]}$ is in the expansion of $\bar t_L^{L[\bar x:\bar l]}$, i.e., $\bar t^{I[\bar x:\bar d]} \in \Exp(\bar t_L^{L[\bar x:\bar l]})$: the proposition follows from the regularity of $\bar t_L^{L[\bar x:\bar l]}$.
    We prove this lemma by induction over the depth of terms. The induction hypothesis is that the lemma holds for depth $k-1$.

    \begin{itemize}
        \item Base case: $t_j$ is a constant $c$.  The variable assignments are irrelevant. Let $c()_L^L  =l$; as $\mul(l)=1$, let $\Exp(l) = \{l\}$. Then $c()^I = l \in \Exp(l)$. %$ = \pi^i(c()_L^L)$ for all $i$.

        \item Base case: $t_j$ is a variable $x$. By Definition~\ref{def:expansion-assignment}, the concrete value assigned to $x$ must be in the expansion of its lifted value.

        \item These are the only terms of depth 1, hence  the induction hypothesis  holds for $k=1$.

        \item Induction step: $t_j$ is of depth $k$ and evaluates to a number. We can apply one of the cases of Theorem~\ref{theorem:modelpreservation} to show that $t_j$ evaluates in $I[\bar x:\bar d]$ to the number in the singleton expansion of $t_{jL}^{L[\bar x:\bar l]}$. Hence $t_j^{I[\bar x:\bar d]} \in \Exp(\bar t_{jL}^{L[\bar x:\bar l]})$.

        \item Induction step: $t_j$ is a term of depth $k$ of the form $f(\bar s)$.
        By the induction hypothesis, for all $n$, $s_n^{I[\bar x:\bar d]} \in \Exp(s_{nL}^{L[\bar x:\bar l]})$.
        Hence, by construction of the expansion of regular function $f$, $t_j^{I[\bar x:\bar d]}=f(\bar s)^{I[\bar x:\bar d]} = \pi^i(f(\bar s_L)^{L[\bar x:\bar l]}) \in \Exp(t_{jL}^{L[\bar x:\bar l]})$.

    \end{itemize}

This covers all cases and completes the proof.
\end{proof}